%% file: CondPathAnalysis.tex
\documentclass[a4paper,12pt,reqno]{amsart}

\usepackage{graphicx,amssymb,datetime,float,MnSymbol,currfile,tikz}
\usepackage[round]{natbib}
\usetikzlibrary{arrows}
\usetikzlibrary{decorations.markings}

\newtheorem{theorem}{Theorem}
\newtheorem{lemma}[theorem]{Lemma}
\newtheorem{corollary}[theorem]{Corollary}

\newtheorem{example}[theorem]{Example}

\def\ci{\!\perp\!}
\def\nci{\!\not\perp\!}
\def\ra{\rightarrow}

\def\la{\leftarrow}

\def\aa{\leftrightarrow}

\def\ao{\leftarrow\!\!\!\!\!\multimap}

\def\oa{\mathrel{\reflectbox{\ensuremath{\ao}}}}

\def\oo{\mathrel{\reflectbox{\ensuremath{\multimap}}}\!\!\!\!\!\multimap}

\newcommand{\comments}[1]{}

\tikzset{tt/.style={decoration={
  markings,
  mark=at position .485 with {\arrow{>}},
  mark=at position .515 with {\arrow{<}}},postaction={decorate}}}

\begin{document}

\title[]{Factorization of the Partial Covariance in Singly-Connected Path Diagrams}

\author[]{Jose M. Pe\~{n}a\\
IDA, Link\"oping University, Sweden\\
jose.m.pena@liu.se}

\date{\currenttime, \ddmmyydate{\today}, \currfilename}

\maketitle

\begin{abstract}
We extend path analysis by showing that, for a singly-connected path diagram, the partial covariance of two random variables factorizes over the nodes and edges in the path between the variables. This result allows us to determine the contribution of each node and edge to the partial covariance. It also allows us to show that Simpson's paradox cannot occur in singly-connected path diagrams.
\end{abstract}

\section{Introduction}

To ease interpretation, linear structural equation models are typically represented as path diagrams: Nodes represent random variables, directed edges represent direct causal relationships, and bidirected edges represent confounding, i.e. correlation between error terms. Moreover, each directed edge is annotated with the corresponding coefficient in the linear structural equation model, a.k.a. path coefficient. Likewise, each bidirected edge is annotated with the corresponding error correlation. A path diagram also brings in computational benefits. For instance, the covariance $\sigma_{XY}$ of two random variables $X$ and $Y$ can be determined from the path diagram. Specifically, $\sigma_{XY}$ can be expressed as the sum over the paths from $X$ to $Y$ of the product of path coefficients and error covariances of the edges in the path \citep{Wright1921,Pearl2009}. Hence, the covariance factorizes over the edges and nodes in the paths. In this work, we develop a similar factorization for the partial covariance $\sigma_{XY \cdot Z}$ in singly-connected path diagrams, i.e. the underlying undirected graph is a tree and, thus, no undirected cycle exists. While path analysis in a singly-connected path diagram determines the contribution of each node and edge to the covariance, our results determine the contribution of each node and edge to the partial covariance. Moreover, we use our results to show that Simpson's paradox cannot occur in singly-connected path diagrams. For path diagrams, Simpson's paradox can be described as the reversal of the sign of the regression coefficient of a random variable $Y$ on a second variable $X$ upon conditioning on a set of variables $Z$ \citep{Pearl2009,Pearl2014}.

Some previous works have certainly studied measures of association for singly-connected path diagrams, or for Gaussian random vectors in general. However, none of these works develop a factorization of the measure of association, as we do in this work. For singly-connected path diagrams, \citet{ChaudhuriandRichardson2003} and \citet{Chaudhuri2005} identify sufficient conditional independencies for ordering some squared partial correlation coefficients. \citet{Chaudhuri2014} extends these results to general Gaussian random vectors. \citet{ChaudhuriandTan2010} report similar general results for absolute values of partial regression coefficients. Finally, \citet{Ong2014} proves similar results for (signed) partial covariances, correlation coefficients and regression coefficients for singly-connected path diagrams and general Gaussian random vectors. In Section \ref{sec:Simpson}, we discuss further the work by Ong. Finally, it should be mentioned that \citet{Pearl2014} identifies three singly-connected path diagrams that cannot lead to Simpson's paradox. Our results are stronger, as we show that Simpson's paradox cannot occur in any singly-connected path diagram.

The rest of this work is structured as follows. Section \ref{sec:nocolliders} presents our factorization of partial covariances for singly-connected path diagrams with no colliders. Section \ref{sec:phenomena} demonstrates our factorization on some examples. Section \ref{sec:colliders} extends the factorization to diagrams with colliders. Section \ref{sec:Simpson} shows that our factorization implies that Simpson's paradox cannot occur in singly-connected path diagrams. Section \ref{sec:discussion} closes with some discussion.

\section{Paths without Colliders}\label{sec:nocolliders}

In this work, we make extensive use of the following recursive definition of the partial covariance of two random variables $X$ and $Y$ given a set of variables $Z$ and a variable $W$ such that $X, Y \notin Z \cup W$ \citep[Section 2.5.3]{Anderson2003}:
\begin{equation}\label{eq:partial}
\sigma_{XY \cdot ZW} =  \sigma_{XY \cdot Z} - \frac{\sigma_{XW \cdot Z} \sigma_{WY \cdot Z}}{\sigma^2_{W \cdot Z}}
\end{equation}
where, for simplicity, we use juxtaposition to denote union. Note that $X$ and $Y$ may be the same random variable, in which case the expression above corresponds to the partial variance $\sigma^2_{X \cdot ZW}=\sigma_{X X \cdot ZW}$. Recall that the partial (co)variances coincide with the conditional (co)variances for Gaussian random vectors.

We continue by recalling the separation criterion for path diagrams \citep[Section 1.2.3]{Pearl2009}. For simplicity, we do not make any distinction between the nodes in the path diagrams and the random variables that they represent. Given a path $\pi_{X:Y}$ from a node $X$ to a node $Y$ in a path diagram, a node $C$ is a collider in $\pi_{X:Y}$ if $A \oa C \ao B$ is a subpath of $\pi_{X:Y}$, where $\oa$ means $\ra$ or $\aa$. Given a set of nodes $Z$, $\pi_{X:Y}$ is said to be $Z$-open if
\begin{itemize}
\item every collider in $\pi_{X:Y}$ is in $Z$ or has some descendant in $Z$, and
\item every non-collider in $\pi_{X:Y}$ is outside $Z$.
\end{itemize}

If there is no $Z$-open path from $X$ to $Y$ (which we denote as $X \ci Y | Z$), then we can readily conclude that $X$ and $Y$ are conditionally independent given $Z$ in the joint normal distribution represented by the path diagram and, thus, $\sigma_{XY \cdot Z}=0$ \citep{Pearl2009}. If on the other hand there is a $Z$-open path from $X$ to $Y$ (which we denote as $X \nci Y | Z$), we assume in this section that it has no colliders, and defer the case with colliders to the next section.

When $X \nci Y | \emptyset$, it is known from path analysis that the covariance $\sigma_{XY}$ of two standardized random variables $X$ and $Y$ can be expressed as the sum for every $\emptyset$-open path from $X$ to $Y$ of the product of the path coefficients and error covariances of the edges in the path \citep{Wright1921,Pearl2009}. For non-standardized variables, one has to multiply the product associated to each path with the variance of the root variable in the path, i.e. the variable with no incoming edges. A path can have no root variables ($X \aa Z \ra \cdots \ra Y$ or $X \la \cdots \la Z \aa W \ra \cdots \ra Y$) or one root variable ($X \ra \cdots \ra Y$ or $X \la \cdots \la Z \ra \cdots \ra Y$).

When $X \nci Y | Z$ with $Z \neq \emptyset$, one may think that $\sigma_{XY \cdot Z}$ can be obtained by first applying path analysis to obtain an expression for $\sigma_{XY}$ and, then, modifying this expression by replacing (co)variances with conditional (co)variances given $Z$. However, this is incorrect as the following example shows.

\begin{figure}
\begin{tabular}{c|c}
\begin{tikzpicture}[inner sep=1mm]
\node at (0,0) (X) {$X$};
\node at (1.5,-1.2) (Z) {$Z$};
\node at (1.5,0) (Y) {$Y$};
\path[->] (X) edge node[above] {$\alpha$} (Y);
\path[->] (Y) edge node[right] {$\delta$} (Z);
\end{tikzpicture}
&
\begin{tikzpicture}[inner sep=1mm]
\node at (0,0) (X) {$X$};
\node at (0,1.2) (eX) {$\epsilon_X$};
\node at (1.5,-1.2) (Z) {$Z$};
\node at (3,-1.2) (eZ) {$\epsilon_Z$};
\node at (1.5,0) (Y) {$Y$};
\node at (3,0) (eY) {$\epsilon_Y$};
\path[->] (X) edge node[above] {$\alpha$} (Y);
\path[->] (Y) edge node[right] {$\delta$} (Z);
\path[->] (eX) edge node[right] {$1$} (X);
\path[->] (eY) edge node[above] {$1$} (Y);
\path[->] (eZ) edge node[above] {$1$} (Z);
\end{tikzpicture}\\
&\\
(i) & (ii)
\end{tabular}\caption{Path diagrams in Example \ref{exa:counterexample}.}\label{fig:counterexample}
\end{figure}

\begin{example}\label{exa:counterexample}
Consider the path diagram (i) in Figure \ref{fig:counterexample}, which corresponds to the following linear structural equation model:
\begin{align*}
X &= \epsilon_X\\
Y &= \alpha X + \epsilon_Y\\
Z &= \delta Y + \epsilon_Z.
\end{align*}
Consider representing the error terms explicitly in the diagram, which results in the path diagram (ii) in Figure \ref{fig:counterexample}. Then,
\[
\sigma_{XY} = cov(X, \alpha X + \epsilon_Y) = \alpha \sigma^2_X + cov(X,\epsilon_Y) = \alpha \sigma^2_X
\]
where the last equality follows from the fact that $cov(X,\epsilon_Y)=0$ since $X \ci \epsilon_Y | \emptyset$. However, \[
\sigma_{XY \cdot Z} = cov(X, \alpha X + \epsilon_Y | Z) = \alpha \sigma^2_{X \cdot Z} + cov(X,\epsilon_Y | Z) \neq \alpha \sigma^2_{X \cdot Z}
\]
where the last inequality follows from the fact that $cov(X,\epsilon_Y | Z) \neq 0$ in general, since $X \nci \epsilon_Y | Z$ because $Z$ is a descendant of $Y$, which is a collider in the path from $X$ to $\epsilon_Y$.
\end{example}

For singly-connected path diagrams, the following two theorems show how to obtain $\sigma_{XY \cdot Z}$ from $\sigma_{XY}$. Interestingly, $\sigma_{XY \cdot Z}$ can still be written as a product over the nodes and edges in the path. See Appendix A for the proofs. Hereinafter, we use the following notation. The parents of a node $X$ are $Pa(X) = \{Y | Y \ra X \}$. The children of $X$ are $Ch(X) = \{Y | X \ra Y \}$. The spouses of $X$ are $Sp(X) = \{Y | X \aa Y \}$.

\begin{theorem}\label{the:condpath1}
Let $\pi_{XY}$ be of the form $X = X_m \la \cdots \la X_2 \la X_1 \ra X_{m+1} \ra \cdots \ra X_{m+n} = Y$ or $X = X_1 \ra X_{2} \ra \cdots \ra X_{m+n} = Y$. Let $Z^i$ be a set of nodes such that each is connected to $Pa(X_i) \cup Sp(X_i)$ by a path that does not contain any node in $\pi_{XY}$.\footnote{It suffices that each node in $Z^i$ is connected to one node in $Pa(X_i) \cup Sp(X_i)$. The connecting path may be of length zero. The path does not need to be open with respect to any set of nodes.} Let $Z_i$ be a set of nodes such that each is connected to $Ch(X_i)$ by a path that does not contain any node in $\pi_{XY}$. Let $Z_i^i=Z_i \cup Z^i$ and $Z_{1:i}^{1:i} = Z_1^1 \cup \cdots \cup Z_i^i$. Then,
\[
\sigma_{X Y \cdot Z_{1:m+n}^{1:m+n}} = \sigma_{X Y} \frac{\sigma^2_{X_1 \cdot Z_1^1}}{\sigma^2_{X_1}} \prod_{i=2}^{m+n} \frac{\sigma^2_{X_i \cdot Z_{1:i-1}^{1:i-1} Z_i^i}}{\sigma^2_{X_i \cdot Z_{1:i-1}^{1:i-1} Z^i}}
\]
where $\sigma_{X Y}$ is obtained by path analysis.
\end{theorem}

\begin{theorem}\label{the:condpath2}
Let $\pi_{XY}$ be of the form $X = X_m \la \cdots \la X_2 \la X_1 \aa X_{m+1} \ra \cdots \ra X_{m+n} = Y$ or $X = X_1 \aa X_{2} \ra \cdots \ra X_{m+n} = Y$. Let $Z^i$ be a set of nodes such that each is connected to $Pa(X_i) \cup Sp(X_i)$ by a path that does not contain any node in $\pi_{XY}$. Let $Z_i$ be a set of nodes such that each is connected to $Ch(X_i)$ by a path that does not contain any node in $\pi_{XY}$. Let $Z_i^i=Z_i \cup Z^i$ and $Z_{1:i}^{1:i} = Z_1^1 \cup \cdots \cup Z_i^i$. Then,
\[
\sigma_{X Y \cdot Z_{1:m+n}^{1:m+n}} = \sigma_{X Y} \prod_{i=1}^{m+n} \frac{\sigma^2_{X_i \cdot Z_{1:i-1}^{1:i-1} Z_i^i}}{\sigma^2_{X_i \cdot Z_{1:i-1}^{1:i-1} Z^i}}
\]
where $\sigma_{X Y}$ is obtained by path analysis, and $Z_{1:0}^{1:0} = \emptyset$.
\end{theorem}

We demonstrate the theorems above on some examples in the next section. Before that, note that the numerator and denominator of the partial variance ratio in the theorems above only differ in that the conditioning set of the former is a superset of the conditioning set of the latter. Thus, the ratio is never greater than 1, since conditioning never increases the variance of a random variable. Therefore, the theorems above show that the partial covariance between two nodes can be computed by multiplying the expression for the covariance given by path analysis with a product of partial variance ratios that account for the reduction of the partial variances of the variables in the path between the two nodes. Thus, like the covariance, the partial covariance factorizes over the nodes and edges in the path. Moreover, the partial covariance is never greater than the covariance. However, both share the same sign, i.e. conditioning does not change the sign of the covariance. This implies that if two nodes $X$ and $Y$ are connected by a path of the form $X \ra \cdots \ra Y$, then conditioning does not change the sign of the regression coefficient of $Y$ on $X$ and, thus, of the causal effect of $X$ on $Y$. This observation will be instrumental in proving in Section \ref{sec:Simpson} that Simpson's paradox does not occur in singly-connected path diagrams. The following corollary is immediate.

\begin{corollary}\label{cor:samesign}
Let $\pi_{XY}$ be of the form in Theorems \ref{the:condpath1} or \ref{the:condpath2}. Moreover, let $\pi_{XY}$ be open with respect to the sets of nodes $U$ and $V$. Then, $sign(\sigma_{XY}) = sign(\sigma_{XY \cdot U})=sign(\sigma_{XY \cdot V})$.
\end{corollary}

The expressions in the theorems above can be simplified by removing irrelevant variables from the conditioning set prior to applying the theorems. Specifically, let $T=Z_{1:m+n}^{1:m+n}$, and let $I=\{I_1,\ldots,I_s\}$ denote all the nodes in $T$ such that $X \cup Y \ci I_i | T \setminus I_i$. Then, $X \cup Y \ci I | T \setminus I$ by repeated application of the intersection property \citep[Proposition 2.1]{Studeny2005} and, thus, $\sigma_{X Y \cdot T}= \sigma_{X Y \cdot T \setminus I}$. In other words, $I$ contains irrelevant nodes. As a matter of fact, $I$ contains all the irrelevant nodes. To see it, assume to the contrary that there exists a second set of nodes $I' \nsubseteq I$ such that $X \cup Y \ci I' | T \setminus I'$. Then, $X \cup Y \ci I_j' | T \setminus I_j'$ for all $I_j' \in I'$ by the weak union property \citep[Lemma 2.1]{Studeny2005}, which contradicts the definition of $I$.

To sum up, the relevance of the theorems above lies in that they somehow complement path analysis: While path analysis in a singly-connected path diagram determines the contribution of each node and edge to the covariance $\sigma_{X Y}$, the theorems above determine the contribution of each node and edge to the partial covariance $\sigma_{X Y \cdot Z_{1:m+n}^{1:m+n}}$. Specifically, the theorems indicate whether the contribution of each node and edge to the covariance changes by conditioning and, if so, by how much. For example, consider Theorem \ref{the:condpath1} and let $\pi_{XY}$ be of the form $X = X_1 \ra X_{2} \ra \cdots \ra X_{m+n} = Y$. It follows from the theorem that the contribution of each edge in $\pi_{XY}$ to the covariance and partial covariance is the same, namely the corresponding path coefficient. It follows from path analysis that the contribution of $X_1$ to the covariance is $\sigma^2_{X_1}$. The theorem shows that this contribution gets reduced by a factor of $\frac{\sigma^2_{X_1 \cdot Z_1^1}}{\sigma^2_{X_1}}$ when conditioning on $Z_1^1$. Likewise, the contribution of $X_2$ to the covariance is $1$. This contribution gets reduced by a factor of $\frac{\sigma^2_{X_i \cdot Z_{1}^{1} Z_2^2}}{\sigma^2_{X_i \cdot Z_{1}^{1} Z^2}}$ when conditioning on $Z_2^2$. This indicates that conditioning on $Z^2$ may change the variance of $X_2$ but it does not constrain $X_2$ so as to alter the contribution of $X_2$. Conditioning on $Z_2$, on the other hand, has the opposite effect. Likewise for the rest of the nodes in $\pi_{XY}$. This fine-grained analysis is not possible with the recursion in Equation \ref{eq:partial}.

\begin{figure}
\begin{tabular}{c|c|c|c}
\begin{tikzpicture}[inner sep=1mm]
\node at (0,0) (X) {$X$};
\node at (1.5,0) (Z) {$Z$};
\node at (3,0) (Y) {$Y$};
\node at (1.5,-1.2) (W) {$W$};
\path[->] (X) edge node[above] {$\alpha$} (Z);
\path[->] (Z) edge node[above] {$\beta$} (Y);
\path[->] (Z) edge node[right] {$\gamma$} (W);
\end{tikzpicture}
&
\begin{tikzpicture}[inner sep=1mm]
\node at (0,0) (X) {$X$};
\node at (1.5,0) (Z) {$Z$};
\node at (3,0) (Y) {$Y$};
\node at (1.5,1.2) (W) {$W$};
\path[->] (X) edge node[above] {$\alpha$} (Z);
\path[->] (Z) edge node[above] {$\beta$} (Y);
\path[->] (W) edge node[right] {$\gamma$} (Z);
\end{tikzpicture}
&
\begin{tikzpicture}[inner sep=1mm]
\node at (0,0) (X) {$X$};
\node at (1.5,-1.2) (Z) {$Z$};
\node at (1.5,0) (Y) {$Y$};
\path[->] (X) edge node[above] {$\alpha$} (Y);
\path[->] (Y) edge node[right] {$\delta$} (Z);
\end{tikzpicture}
&
\begin{tikzpicture}[inner sep=1mm]
\node at (0,0) (X) {$X$};
\node at (0,-1.2) (Z) {$Z$};
\node at (1.5,0) (Y) {$Y$};
\path[->] (X) edge node[above] {a} (Y);
\path[->] (X) edge node[right] {b} (Z);
\end{tikzpicture}\\
&&&\\
(i) & (ii) & (iii) & (iv)
\end{tabular}\caption{Path diagrams in Examples \ref{exa:fig2and3} and \ref{exa:fig4and5}.}\label{fig:examples}
\end{figure}

\section{Causal Phenomena Explained}\label{sec:phenomena}

In this section, we demonstrate Theorems \ref{the:condpath1} and \ref{the:condpath2} on some examples that shed light on some causal phenomena. The examples are borrowed from \citet{Pearl2013}, who studied them using Equation \ref{eq:partial}. The objective of this section is purely illustrative. That is, we do not compare our explanations and those by \citet{Pearl2013}, as our theorems and Equation \ref{eq:partial} address different problems.

\begin{example}\label{exa:fig2and3}
Consider the path diagram (i) in Figure \ref{fig:examples}. The causal effect of $X$ on $Y$ is given by the regression coefficient $r_{YX} = \alpha \beta$. Since $W$ does not lie on the causal path from $X$ to $Y$, one may think that the causal effect of $X$ on $Y$ is also given by the partial regression coefficient $r_{YX \cdot W}$, which can be computed from the subpopulation satisfying $W=w$ for any $w$. However, this is incorrect as shown by \citet[Section 3.2]{Pearl2013}. We arrive at the same conclusion as \citet{Pearl2013} by applying Theorem \ref{the:condpath1} with $X_1=X, X_2=Z, X_3=Y, Z_1^1=Z^2=Z^3=Z_3^3=\emptyset$, and $Z^2_2=\{W\}$, which gives that
\[
\sigma_{XY \cdot W} = \sigma_{XY} \frac{\sigma^2_X}{\sigma^2_X} \frac{\sigma^2_{Z \cdot W}}{\sigma^2_Z} \frac{\sigma^2_{Y \cdot W}}{\sigma^2_{Y \cdot W}}.
\]
Moreover, $\sigma_{XY} = \sigma^2_X \alpha \beta$ by path analysis. Then,
\[
r_{YX \cdot W} = \frac{\sigma_{XY \cdot W}}{\sigma^2_{X \cdot W}} = \alpha \beta \frac{\sigma^2_X}{\sigma^2_{X \cdot W}} \frac{\sigma^2_{Z \cdot W}}{\sigma^2_Z}
\]
and, thus, $r_{YX \cdot W} \neq \alpha \beta$ unless $\gamma=0$ or $\alpha=\sigma_Z / \sigma_X$. To see it, note that
\[
\sigma^2_{X \cdot W} = \sigma^2_{X} - \frac{\sigma_{XW} \sigma_{WX}}{\sigma^2_{W}} = \sigma^2_{X} - \frac{( \sigma^2_X \alpha \gamma )^2}{\sigma^2_{W}} = \sigma^2_{X} \Big(\frac{\sigma^2_{W} - \sigma^2_X \alpha^2 \gamma^2}{\sigma^2_{W}} \Big)
\]
and, similarly,
\[
\sigma^2_{Z \cdot W} = \sigma^2_{Z} \Big(\frac{\sigma^2_{W} - \sigma^2_{Z} \gamma^2}{\sigma^2_{W}} \Big).
\]
Then,
\[
\frac{\sigma^2_X}{\sigma^2_{X \cdot W}} \frac{\sigma^2_{Z \cdot W}}{\sigma^2_Z} = \frac{\sigma^2_{W} - \sigma^2_{Z} \gamma^2}{\sigma^2_{W} - \sigma^2_{X} \alpha^2 \gamma^2} = 1
\]
if and only if $\gamma=0$ or $\alpha=\sigma_Z / \sigma_X$.\footnote{The effect of setting $\gamma=0$ on $r_{YX \cdot W}$ is as follows. Setting $\gamma=0$ is equivalent to removing the edge $Z \ra W$ from the path diagram (i) in Figure \ref{fig:examples}, which implies that $r_{YX \cdot W} = \alpha \beta$. The effect of setting $\alpha=\sigma_Z / \sigma_X$ on $r_{YX \cdot W}$ is as follows. The path diagram (i) in Figure \ref{fig:examples} corresponds to a model that contains the linear structural equation $Z=\alpha X + \epsilon_Z$ with $X \ci \epsilon_Z | \emptyset$. Then, $\sigma_Z^2=\alpha^2 \sigma_X^2 + var(\epsilon_Z)$ and, thus, $var(\epsilon_Z)=0$ when $\alpha=\sigma_Z / \sigma_X$, i.e. $Z$ is completely determined by $X$. In other words, the diagram (i) in Figure \ref{fig:examples} reduces to the diagram (iv), which is studied in Example \ref{exa:fig4and5}.}

As also shown by \citet[Section 3.2]{Pearl2013}, no bias is introduced when conditioning on $W$ in the path diagram (ii) in Figure \ref{fig:examples}. We arrive at the same conclusion as \citet{Pearl2013} by applying Theorem \ref{the:condpath1} with $X_1=X, X_2=Z, X_3=Y, Z_1^1=Z_3^3=\emptyset$, and $Z^2=Z_2^2=\{W\}$, which gives that
\[
\sigma_{XY \cdot W} = \sigma_{XY} \frac{\sigma^2_X}{\sigma^2_X} \frac{\sigma^2_{Z \cdot W}}{\sigma^2_{Z \cdot W}} \frac{\sigma^2_{Y \cdot W}}{\sigma^2_{Y \cdot W}}.
\]
Moreover, $\sigma_{XY} = \sigma^2_X \alpha \beta$ by path analysis. Then,
\[
r_{YX \cdot W} = \frac{\sigma_{XY \cdot W}}{\sigma^2_{X \cdot W}} = \alpha \beta \frac{\sigma^2_X}{\sigma^2_{X \cdot W}} = \alpha \beta
\]
where the last equality follows from the fact that $X \ci W | \emptyset$ and, thus, $\sigma^2_X = \sigma^2_{X \cdot W}$.
\end{example}

\begin{example}\label{exa:fig4and5}
Consider the path diagram (iii) in Figure \ref{fig:examples}. The causal effect of $X$ on $Y$ is given by $r_{YX} = \alpha$. As shown by \citet[Section 3.3]{Pearl2013}, conditioning on $Z$ introduces a bias. We arrive at the same conclusion as \citet{Pearl2013} by applying Theorem \ref{the:condpath1} with $X_1=X, X_2=Y, Z_1^1=Z^2=\emptyset$, and $Z^2_2=\{Z\}$, which gives that
\[
\sigma_{XY \cdot Z} = \sigma_{XY} \frac{\sigma^2_X}{\sigma^2_X} \frac{\sigma^2_{Y \cdot Z}}{\sigma^2_Y}.
\]
Moreover, $\sigma_{XY} = \sigma^2_X \alpha$ by path analysis. Then,
\begin{equation}\label{eq:partial2}
r_{YX \cdot Z} = \frac{\sigma_{XY \cdot Z}}{\sigma^2_{X \cdot Z}} = \alpha \frac{\sigma^2_X}{\sigma^2_{X \cdot Z}} \frac{\sigma^2_{Y \cdot Z}}{\sigma^2_Y}.
\end{equation}
and, thus, $r_{YX \cdot Z} \neq \alpha$ unless $\delta=0$ or $\alpha=\sigma_Y / \sigma_X$ as shown in Example \ref{exa:fig2and3}. In summary, the causal effect of $X$ on $Y$ cannot be computed from the subpopulation satisfying $Z=z$ because $r_{YX \cdot Z} \neq \alpha$. However, if $\sigma^2_{X}$ and $\sigma^2_{Z}$ are known, then the causal effect can be computed from that subpopulation by correcting $r_{YX \cdot Z}$ as shown in Equation \ref{eq:partial2}.

As also shown by \citet[Section 3.3]{Pearl2013}, no bias is introduced when conditioning on $Z$ in the path diagram (iv) in Figure \ref{fig:examples}. We arrive at the same conclusion as \citet{Pearl2013} by applying Theorem \ref{the:condpath1} with $X_1=X, X_2=Y, Z_1^1=\{Z\}$, and $Z^2=Z_2^2=\emptyset$, which gives that
\[
\sigma_{XY \cdot Z} = \sigma_{XY} \frac{\sigma^2_{X \cdot Z}}{\sigma^2_X} \frac{\sigma^2_{Y \cdot Z}}{\sigma^2_{Y \cdot Z}}.
\]
Moreover, $\sigma_{XY} = \sigma^2_X a$ by path analysis. Then,
\[
r_{YX \cdot Z} = \frac{\sigma_{XY \cdot Z}}{\sigma^2_{X \cdot Z}} = a.
\]
\end{example}

The examples above show that conditioning on a child of a mediator or on a child of the effect introduces a bias in the estimation of the causal effect of interest. Appendix B illustrates with experiments how this bias may lead to suboptimal decision making. On the other hand, the examples above show that conditioning on a parent of a mediator or on a child of the cause does not introduce any bias, which implies that the causal effect of interest can be estimated from a sample of the corresponding subpopulation.

For completeness, we show below that conditioning on a parent of the cause or on a parent of the effect does not introduce any bias.

\begin{example}\label{exa:exa6}
Consider the path diagram (ii) in Figure \ref{fig:examples}. The causal effect of $Z$ on $Y$ is given by $r_{YZ} = \beta$. We conclude that $r_{YZ \cdot W}=\beta$ by applying Theorem \ref{the:condpath1} with $X_1=Z, X_2=Y, Z_1^1=\{W\}$, and $Z^2=Z^2_2=\emptyset$. Specifically,
\[
\sigma_{ZY \cdot W} = \sigma_{ZY} \frac{\sigma^2_{Z \cdot W}}{\sigma^2_Z} \frac{\sigma^2_{Y \cdot W}}{\sigma^2_{Y \cdot W}}.
\]
Moreover, $\sigma_{ZY} = \sigma^2_Z \beta$ by path analysis. Then,
\[
r_{YZ \cdot W} = \frac{\sigma_{ZY \cdot W}}{\sigma^2_{Z \cdot W}} = \beta.
\]
This result also follows from the first rule of do-calculus \citep[Section 3.4]{Pearl2009}.

Consider again the path diagram (ii) in Figure \ref{fig:examples}. The causal effect of $X$ on $Z$ is given by $r_{ZX} = \alpha$. We conclude that $r_{ZX \cdot W}=\alpha$ by applying Theorem \ref{the:condpath1} with $X_1=X, X_2=Z, Z_1^1=\emptyset$, and $Z^2=Z^2_2=\{W\}$. Specifically,
\[
\sigma_{XZ \cdot W} = \sigma_{XZ} \frac{\sigma^2_{X}}{\sigma^2_X} \frac{\sigma^2_{Z \cdot W}}{\sigma^2_{Z \cdot W}}.
\]
Moreover, $\sigma_{XZ} = \sigma^2_X \alpha$ by path analysis. Then,
\[
r_{ZX \cdot W} = \frac{\sigma_{XZ \cdot W}}{\sigma^2_{X \cdot W}} = \alpha \frac{\sigma^2_X}{\sigma^2_{X \cdot W}} = \alpha
\]
where the last equality follows from the fact that $X \ci W | \emptyset$ and, thus, $\sigma^2_X = \sigma^2_{X \cdot W}$. We can arrive at the same conclusion by applying the definition of partial covariance. Specifically,
\[
\sigma_{XZ \cdot W} =  \sigma_{XZ} - \frac{\sigma_{XW} \sigma_{WZ}}{\sigma^2_{W}} = \sigma_{XZ}
\]
because $X \ci W | \emptyset$ implies that $\sigma_{XW} = 0$.

\end{example}

\begin{figure}
\begin{tabular}{c|c}
\begin{tikzpicture}[inner sep=1mm]
\node at (0,0) (X) {$X$};
\node at (2.5,1.2) (Z) {$Z$};
\node at (2,0) (Y) {$Y$};
\node at (1,1.2) (U) {$U$};
\path[->] (X) edge node[above] {$\alpha$} (Y);
\path[->] (U) edge node[left] {$\beta$} (X);
\path[->] (U) edge node[right] {$\gamma$} (Y);
\path[->] (U) edge node[above] {$\delta$} (Z);
\end{tikzpicture}
&
\begin{tikzpicture}[inner sep=1mm]
\node at (0,0) (X) {$X$};
\node at (2.5,1.2) (Z) {$Z$};
\node at (2,0) (Y) {$Y$};
\node at (1,1.2) (U) {$U$};
\path[->] (X) edge node[above] {$\alpha$} (Y);
\path[->] (U) edge node[left] {$\beta$} (X);
\path[->] (U) edge node[right] {$\gamma$} (Y);
\path[<-] (U) edge node[above] {$\delta$} (Z);
\end{tikzpicture}\\
&\\
(i) & (ii)
\end{tabular}\caption{Path diagrams in Example \ref{exa:fig13}.}\label{fig:fig13}
\end{figure}
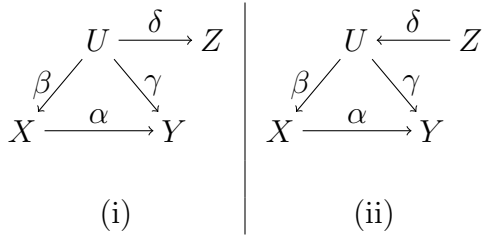

Finally, the example below shows that Theorems \ref{the:condpath1} and \ref{the:condpath2} may be of help even when the path diagram at hand is not singly-connected. The example is borrowed from \citet{Pearl2013}.

\begin{example}\label{exa:fig13}
Consider the path diagram (i) in Figure \ref{fig:fig13}. Let us denote it by $G$. Moreover, let $G^{\alpha}$ denote the diagram that results when the edge $X \ra Y$ is deleted from $G$. Since $X \ci Y | U$ holds in $G^{\alpha}$, we have that $\alpha = r_{YX \cdot U}$ \citep[Theorem 5.3.1]{Pearl2009}. However, if $U$ is unobserved then $r_{YX \cdot U}$ cannot be computed. Assume that the proxy $Z$ of $U$ is observed and, thus, $r_{YX \cdot Z}$ can be computed. Of course, $\alpha \neq r_{YX \cdot Z}$ because $X \ci Y | Z$ does not hold in $G^{\alpha}$. However, \citet[Section 3.11]{Pearl2013} shows that the bias introduced by adjusting for $Z$ instead of $U$ vanishes as the correlation between $U$ and $Z$ grows, i.e. when $Z$ is a good proxy of $U$. The same occurs in the path diagram (ii) in Figure \ref{fig:fig13}.

Although the path diagrams in Figure \ref{fig:fig13} are not singly-connected, we can still use our results to reach the same conclusions as Pearl. Since the covariance of $X$ and $Y$ may differ in $G$ and $G^{\alpha}$, we use $\sigma_{XY}$ for the former and $\sigma_{XY}^{\alpha}$ for the latter. For the same reason, we distinguish between $\sigma_{XY \cdot Z}$ and $\sigma_{XY \cdot Z}^{\alpha}$. Since the variance of $U$ is the same in $G$ and $G^{\alpha}$, we simply denote it as $\sigma_{U}^2$. For the same reason, we use $\sigma_{U \cdot Z}^2$ to denote the partial variance of $U$ given $Z$ in both $G$ and $G^{\alpha}$. Note that checking whether $X \ci Y | Z$ holds in $G^{\alpha}$ is equivalent to checking whether $\sigma_{XY \cdot Z}^{\alpha}=0$ holds. Since $G^{\alpha}$ is a singly-connected path diagram, we can apply Theorem \ref{the:condpath1} and conclude that $\sigma_{XY \cdot Z}^{\alpha} = \sigma_{XY}^{\alpha} \sigma_{U \cdot Z}^2 / \sigma_{U}^2$. This implies that, although conditioning on $Z$ does not nullify the covariance of $X$ and $Y$ in $G^{\alpha}$, it does reduce it. Moreover, the greater the correlation between $U$ and $Z$, the greater the reduction and, thus, the closer $r_{YX \cdot Z}$ comes to $\alpha$. We illustrate this with some experiments in Appendix C.
\end{example}

\section{Paths with Colliders}\label{sec:colliders}

In this section, we address the case where $\pi_{XY}$ has colliders. Specifically, let $\pi_{XY}$ be $Z$-open. Given a collider $C$ in $\pi_{XY}$, an opener is any node $W \in Z$ such that $C = C_1 \ra \cdots \ra C_n = W$ and $C_1, \ldots, C_{n-1} \notin Z$. Note that $C$ is an opener if $C \in Z$.

\begin{theorem}\label{the:collider}
Let $C$ be a collider in $\pi_{XY}$. Moreover, let $\pi_{XY}$ be closed with respect to $Z$ but open with respect to $Z \cup Z_{1:n}^{1:n} \cup W_{1:n}$ where (i) $W_1, \ldots, W_n$ are openers for $C$, (ii) $Z^i$ is a the set of nodes such that each is connected to $Pa(W_i) \cup Sp(W_i)$ by a path that does not contain any node in $\pi_{XY}$ or $\pi_{CW_i}$, (iii) $Z_i$ is a set of nodes such that each is connected to $Ch(W_i)$ by a path, and (iv) $Z_{1:i}^{1:j} = Z_1 \cup \cdots \cup Z_i \cup Z^1 \cup \cdots \cup Z^j$. Then,
\begin{equation}\label{eq:collider}
\sigma_{XY \cdot Z Z_{1:n}^{1:n} W_{1:n}} = - \sum_{i=1}^n \frac{\sigma_{X W_i \cdot Z Z_{1:i-1}^{1:i} W_{1:i-1}} \sigma_{W_i Y \cdot Z Z_{1:i-1}^{1:i} W_{1:i-1}}}{\sigma^2_{W_i \cdot Z Z_{1:i-1}^{1:i} W_{1:i-1}}}
\end{equation}
where $Z_{1:0}^{1:1} = Z^1$ and $W_{1:0}=\emptyset$.
\end{theorem}

In the theorem above, if $\pi_{X W_i}$ has some collider then $\sigma_{X W_i \cdot Z Z_{1:i-1}^{1:i} W_{1:i-1}}$ is obtained by recursively applying the theorem to $\pi_{X W_i}$. When $\pi_{X W_i}$ has no colliders, $\sigma_{X W_i \cdot Z Z_{1:i-1}^{1:i} W_{1:i-1}}$ is obtained as shown in Theorems \ref{the:condpath1} and \ref{the:condpath2}. Likewise for $\pi_{W_i Y}$ and $\sigma_{W_i Y \cdot Z Z_{1:i-1}^{1:i} W_{1:i-1}}$. Example \ref{exa:collider} demonstrates this recursive procedure. Specifically, let $\pi_{XY}$ have colliders $C_1, \ldots, C_k$, where $C_i$ has openers $\mathcal{W}_i = \{W_{i1}, \ldots, W_{i n_i}\}$. Then, the recursive procedure just described allows us to write $\sigma_{XY \cdot Z Z_{1:n}^{1:n} W_{1:n}}$ as
\begin{equation}\label{eq:recursion}
(-1)^k \sum_{O_1 \in \mathcal{W}_1} \cdots \sum_{O_k \in \mathcal{W}_k} \frac{\sigma_{X O_1 \cdot U_{O_1}} \sigma_{O_1 O_2 \cdot U_{O_{1:2}}} \cdots \: \sigma_{O_{k-1} O_k \cdot U_{O_{1:k}}} \sigma_{O_k Y \cdot U_{O_{1:k}}}}{\sigma^2_{O_1 \cdot U_{O_1}} \sigma^2_{O_2 \cdot U_{O_{1:2}}} \cdots \: \sigma^2_{O_{k-1} \cdot U_{O_{1:k}}} \sigma^2_{O_k \cdot U_{O_{1:k}}}}
\end{equation}
for some sets of nodes $U_{O_1}, U_{O_{1:2}}, \ldots, U_{O_{1:k}}$. In other words, the partial covariance decomposes as a sum over the different ways of opening $\pi_{XY}$, and each term in the sum is a product of calls to Theorems \ref{the:condpath1} and \ref{the:condpath2}. Then, each term in the sum factorizes over the nodes and edges of $\pi_{XY}$. This resembles how path analysis on unconstrained path diagrams decomposes the covariance of two random variables over the different $\emptyset$-open paths between them. We demonstrate the theorem above with an example.

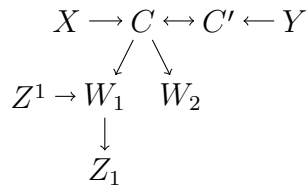
\begin{figure}
\begin{tikzpicture}[inner sep=1mm]
\node at (0,0) (X) {$X$};
\node at (1,0) (C) {$C$};
\node at (2,0) (Cp) {$C'$};
\node at (.5,-1) (W1) {$W_1$};
\node at (1.5,-1) (W2) {$W_2$};
\node at (3,0) (Y) {$Y$};
\node at (-.5,-1) (Zp) {$Z^1$};
\node at (.5,-2) (Zc) {$Z_1$};
\path[->] (X) edge (C);
\path[<->] (C) edge (Cp);
\path[->] (C) edge (W1);
\path[->] (C) edge (W2);
\path[->] (Y) edge (Cp);
\path[->] (W1) edge (Zc);
\path[->] (Zp) edge (W1);
\end{tikzpicture}\caption{Path diagram in Example \ref{exa:collider}.}\label{fig:collider}
\end{figure}

\begin{example}\label{exa:collider}
Consider the path diagram in Figure \ref{fig:collider}. Then, the partial covariance $\sigma_{XY \cdot C' Z_1^1 W_{1:2}}$ can be computed with the help of Theorem \ref{the:collider} with $Z=\{C'\}$. Specifically,
\[
\sigma_{XY \cdot C' Z_1^1 W_{1:2}} = - \frac{\sigma_{X W_1 \cdot C' Z^1} \sigma_{W_1 Y \cdot C' Z^1}}{\sigma^2_{W_1 \cdot C' Z^1}} - \frac{\sigma_{X W_2 \cdot C' Z^1_1 W_1} \sigma_{W_2 Y \cdot C' Z^1_1 W_1}}{\sigma^2_{W_2 \cdot C' Z^1_1 W_1}}.
\]
Moreover, $\sigma_{X W_1 \cdot C' Z^1}$ and $\sigma_{X W_2 \cdot C' Z^1_1 W_1}$ can be computed as shown in Theorem \ref{the:condpath1}. On the other hand, $\sigma_{W_1 Y \cdot C' Z^1}$ and $\sigma_{W_2 Y \cdot C' Z^1_1 W_1}$ can be computed by applying Theorem \ref{the:collider} again with $Z=\{Z^1\}$ and $Z=Z^1_1 \cup \{W_1\}$, respectively. Specifically,
\[
\sigma_{W_1 Y \cdot C' Z^1} = - \frac{\sigma_{W_1 C' \cdot Z^1} \sigma_{C' Y \cdot Z^1}}{\sigma^2_{C' \cdot Z^1}}
\]
and
\[
\sigma_{W_2 Y \cdot C' Z^1_1 W_1} = - \frac{\sigma_{W_2 C' \cdot Z^1_1 W_1} \sigma_{C' Y \cdot Z^1_1 W_1}}{\sigma^2_{C' \cdot Z^1_1 W_1}}
\]
where the partial covariances in the numerators can be computed as shown in Theorems \ref{the:condpath1} and \ref{the:condpath2}. Putting all together, we have that
\[
\sigma_{XY \cdot C' Z_1^1 W_{1:2}} = \frac{\sigma_{X W_1 \cdot C' Z^1} \sigma_{W_1 C' \cdot Z^1} \sigma_{C' Y \cdot Z^1}}{\sigma^2_{W_1 \cdot C' Z^1} \sigma^2_{C' \cdot Z^1}} + \frac{\sigma_{X W_2 \cdot C' Z^1_1 W_1} \sigma_{W_2 C' \cdot Z^1_1 W_1} \sigma_{C' Y \cdot Z^1_1 W_1}}{\sigma^2_{W_2 \cdot C' Z^1_1 W_1} \sigma^2_{C' \cdot Z^1_1 W_1}}
\]
which confirms Equation \ref{eq:recursion} and the discussion thereof.
\end{example}

\section{Simpson's Paradox}\label{sec:Simpson}

In this section, we use the theorems developed in the previous sections to show that Simpson's paradox does not occur in singly-connected path diagrams. For path diagrams, Simpson's paradox can be described as the reversal of the sign of the regression coefficient of a random variable $Y$ on a second variable $X$ upon conditioning on a set of variables $Z$. Note that this is a generalization of the definition by \citet{Pearl2013,Pearl2014}, who restricts $Z$ to a singleton. \citet[Section 3.1]{Pearl2013} shows that the paradox can well occur for the path diagram $X \ra Y \la Z \ra X$. Note that the diagram is not singly-connected. \citet[Section 2.2]{Pearl2014} argues that the paradox does not occur for the singly-connected path diagrams $Z \la X \ra Y$, $Z \ra X \ra Y$, and $X \ra Y \la Z$, because the association between $X$ and $Y$ is collapsible over $Z$. However, the correctness of this statement depends on the definition of association. To see it, recall from \citet[Definition 6.5.1]{Pearl2009} that given a functional $g(p(x,y))$ that measures the association between two random variables $Y$ and $X$ in $p(x,y)$, we say that $g$ is collapsible over a variable $Z$ if
\[
E_z[g(p(x,y|z))] = g(p(x,y)).
\]
If we now consider the diagram $Z \la X \ra Y$ and let $g$ be the covariance between $Y$ and $X$, then collapsibility does not hold since
\begin{align*}
E_z[g(p(x,y|z))] = E_z[cov(X,Y|Z=z)] &= cov(X,Y|Z)\\
& = \sigma_{XY \cdot Z} \neq \sigma_{XY} = g(p(x,y))
\end{align*}
where the second equality follows from the fact that the conditional covariance does not depend on the value on which we condition, and the inequality is proven in Example \ref{exa:fig4and5}. Similarly for the diagram $Z \ra X \ra Y$ as shown in Example \ref{exa:exa6}. For the diagram $X \ra Y \la Z$, on the other hand, collapsibility does hold as shown in Example \ref{exa:exa6}. If we instead let $g$ be the regression coefficient of $Y$ on $X$, then collapsibility holds for the three diagrams under consideration, as shown in Examples \ref{exa:fig4and5} and \ref{exa:exa6}. Moreover, \citet{Pearl2014} does not discuss if Simpson's paradox can occur for the diagram $X \ra Y \ra Z$. Recall that Example \ref{exa:fig4and5} shows that collapsibility does not hold for this diagram, regardless of whether association means covariance or regression coefficient. Pearl does not discuss either the case of singly-connected path diagrams where $X$ and $Y$ are connected by a path of length greater than one with and without colliders, or the case where the conditioning set contains more than one variable. We fill these gaps below.

Note that Simpson's paradox concerns the sign of the regression coefficient of a random variable $Y$ on a random variable $X$ upon conditioning on a set of variables $Z$ or, equivalently, it concerns the sign of the covariance between $X$ and $Y$ upon conditioning on $Z$. Therefore, we are interested in the collapsibility of the sign rather than in the collapsibility of the regression coefficient or covariance. Corollary \ref{cor:samesign} implies that conditioning does not change the sign of the covariance for paths without colliders. The following theorem shows that this also holds for paths with colliders. Consequently, Simpson's paradox cannot occur in any singly-connected path diagram.

\begin{theorem}\label{the:Simpson}
Let $\pi_{XY}$ be open with respect to the sets of nodes $U$ and $V$. Then, $sign(\sigma_{XY \cdot U})=sign(\sigma_{XY \cdot V})$.
\end{theorem}

Note that the result above is actually stronger than required to disprove Simpson's paradox, because $U$ may neither include nor be included in $V$.

It is worth mentioning that \citet[Theorem 6.6]{Ong2014} proves the theorem above by other means when $U$ and $V$ only contain descendants of the colliders in $\pi_{XY}$. Ong states that the proof can be extended to include other conditionates, as in our theorem. However, since he does not provide the details, we believe that our proof fills some gap. Moreover, Ong does not discuss the relevance of this result for disproving Simpson's paradox. Therefore, our discussion above fills some gap, too.

\section{Discussion}\label{sec:discussion}

In this work, we have extended path analysis by showing that, for a singly-connected path diagram, the partial covariance of two random variables factorizes over the nodes and edges in the path between the variables. This result applies even when the path contains colliders. We find the case where the path has no colliders particularly interesting, since then the partial covariance can be computed by multiplying the expression for the covariance given by path analysis with a product of partial variance ratios that account for the reduction of the partial variances of the variables in the path. Moreover, these results have allowed us to show that Simpson's paradox cannot occur in singly-connected path diagrams. Naturally, we would like in the future to extend our results beyond singly-connected path diagrams. Appendix E presents a first attempt in this direction.

\section*{Acknowledgments}

We thank the Reviewers for their comments, which helped us to improve our work.

\section*{Appendix A: Proofs of Sections \ref{sec:nocolliders}-\ref{sec:Simpson}}

Recall that in all the results in this appendix the path diagram is assumed to be singly-connected.

\begin{lemma}\label{lem:root}
Let $S$ be the root node in a path $\pi_{XY}$ without colliders, i.e. $A \la S \ra B$ or $S \ra B$ is a subpath of $\pi_{XY}$. Note that $S=X$ or $S=Y$ in the latter case. Let $W$ be a set of nodes such that each is connected to $Pa(S) \cup Ch(S) \cup Sp(S)$ by a path that does not contain any node in $\pi_{XY}$. Then,
\[
\sigma_{XY \cdot ZW} =  \sigma_{XY \cdot Z} \frac{\sigma^2_{S \cdot ZW}}{\sigma^2_{S \cdot Z}}.
\]
\end{lemma}

\begin{proof}
Assume that $W$ is a singleton. Consider first the case where $A \la S \ra B$ is a subpath of $\pi_{XY}$. Note that $X \ci W | Z \cup S$. Then,
\[
0 = \sigma_{XW \cdot ZS} = \sigma_{XW \cdot Z} - \frac{\sigma_{XS \cdot Z} \sigma_{SW \cdot Z}}{\sigma^2_{S \cdot Z}}
\]
which implies that $\sigma_{XW \cdot Z} = \delta_{XS \cdot Z} \sigma_{SW \cdot Z}$ where $\delta_{XS \cdot Z} = \sigma_{XS \cdot Z} / \sigma^2_{S \cdot Z}$. Likewise, $Y \ci W | Z \cup S$ implies that $\sigma_{YW \cdot Z} = \delta_{YS \cdot Z} \sigma_{SW \cdot Z}$ where $\delta_{YS \cdot Z} = \sigma_{YS \cdot Z} / \sigma^2_{S \cdot Z}$. Likewise, $X \ci Y | Z \cup S$ implies that
\[
0 = \sigma_{XY \cdot ZS} = \sigma_{XY \cdot Z} - \frac{\sigma_{XS \cdot Z} \sigma_{SY \cdot Z}}{\sigma^2_{S \cdot Z}}
\]
which implies that $\sigma_{XY \cdot Z} = \delta_{XS \cdot Z} \delta_{YS \cdot Z} \sigma^2_{S \cdot Z}$. Therefore,
\begin{align*}
\sigma_{XY \cdot ZW} &= \sigma_{XY \cdot Z} - \frac{\sigma_{XW \cdot Z} \sigma_{WY \cdot Z}}{\sigma^2_{W \cdot Z}}\\
& = \delta_{XS \cdot Z} \delta_{YS \cdot Z} \sigma^2_{S \cdot Z} - \frac{\delta_{XS \cdot Z} \sigma_{SW \cdot Z} \delta_{YS \cdot Z} \sigma_{SW \cdot Z}}{\sigma^2_{W \cdot Z}}\\
& = \delta_{XS \cdot Z} \delta_{YS \cdot Z} \Big( \sigma^2_{S \cdot Z} - \frac{\sigma_{SW \cdot Z} \sigma_{SW \cdot Z}}{\sigma^2_{W \cdot Z}} \Big)\\
& = \delta_{XS \cdot Z} \delta_{YS \cdot Z} \sigma^2_{S \cdot ZW} =  \sigma_{XY \cdot Z} \frac{\sigma^2_{S \cdot ZW}}{\sigma^2_{S \cdot Z}}.
\end{align*}

Now, consider the case where $S \ra B$ is a subpath of $\pi_{XY}$. Assume without loss of generality that $S=X$. Note that $Y \ci W | Z \cup X$. Then,
\[
0 = \sigma_{YW \cdot ZX} = \sigma_{YW \cdot Z} - \frac{\sigma_{YX \cdot Z} \sigma_{XW \cdot Z}}{\sigma^2_{X \cdot Z}}
\]
which implies that
\[
\sigma_{YW \cdot Z} = \frac{\sigma_{YX \cdot Z} \sigma_{XW \cdot Z}}{\sigma^2_{X \cdot Z}}.
\]
Therefore,
\begin{align*}
\sigma_{XY \cdot ZW} &= \sigma_{XY \cdot Z} - \frac{\sigma_{XW \cdot Z} \sigma_{WY \cdot Z}}{\sigma^2_{W \cdot Z}} = \sigma_{XY \cdot Z} - \frac{\sigma_{XW \cdot Z} \sigma_{YX \cdot Z} \sigma_{XW \cdot Z}}{\sigma^2_{W \cdot Z} \sigma^2_{X \cdot Z}}\\
&= \sigma_{XY \cdot Z} \Big( 1 - \frac{\sigma_{XW \cdot Z} \sigma_{XW \cdot Z}}{\sigma^2_{W \cdot Z} \sigma^2_{X \cdot Z}} \Big) = \frac{\sigma_{XY \cdot Z}}{\sigma^2_{X \cdot Z}} \Big( \sigma^2_{X \cdot Z} - \frac{\sigma_{XW \cdot Z} \sigma_{XW \cdot Z}}{\sigma^2_{W \cdot Z}} \Big)\\
&= \sigma_{XY \cdot Z} \frac{\sigma^2_{X \cdot ZW}}{\sigma^2_{X \cdot Z}}.
\end{align*}

Repeated application of the paragraphs above proves the result for when $W$ is a set. Specifically, let $W=\{W_1, \ldots, W_n\}$. Then,
\[
\sigma_{XY \cdot ZW_1} =  \sigma_{XY \cdot Z} \frac{\sigma^2_{S \cdot ZW_1}}{\sigma^2_{S \cdot Z}}
\]
by replacing $W$ with $W_1$ in the paragraphs above. Likewise, 
\[
\sigma_{XY \cdot ZW_1W_2} =  \sigma_{XY \cdot ZW_1} \frac{\sigma^2_{S \cdot ZW_1W_2}}{\sigma^2_{S \cdot ZW_1}}
\]
by replacing $Z$ and $W$ with $Z \cup \{W_1\}$ and $W_2$, respectively, in the paragraphs above. These last two results imply that 
\[
\sigma_{XY \cdot ZW_1W_2} =  \sigma_{XY \cdot Z} \frac{\sigma^2_{S \cdot ZW_1W_2}}{\sigma^2_{S \cdot Z}}.
\]
Continuing with this process for $W_3, \ldots, W_n$ yields the desired result.
\end{proof}

\begin{lemma}\label{lem:nonrootpasp}
Let $S$ be a non-root node in a path $\pi_{XY}$ without colliders, i.e. $A \oa S \ra B$ or $A \oa S$ is a subpath of $\pi_{XY}$. Note that $S=X$ or $S=Y$ in the latter case. Let $W$ be a set of nodes such that each is connected to $Pa(S) \cup Sp(S)$ by a path that does not contain any node in $\pi_{XY}$. Then,
\[
\sigma_{XY \cdot ZW} =  \sigma_{XY \cdot Z}
\]
if $Z$ contains no descendants of $S$.
\end{lemma}

\begin{proof}
Assume that $W$ is a singleton. Then,
\begin{align*}
\sigma_{XY \cdot ZW} &= \sigma_{XY \cdot Z} - \frac{\sigma_{XW \cdot Z} \sigma_{WY \cdot Z}}{\sigma^2_{W \cdot Z}}
\end{align*}
which implies that $\sigma_{XY \cdot ZW} = \sigma_{XY \cdot Z}$ because $\sigma_{XW \cdot Z}=0$ or $\sigma_{WY \cdot Z}=0$ since $X \ci W | Z$ or $W \ci Y | Z$.

Repeated application of the paragraph above proves the result for when $W$ is a set. Specifically, let $W=\{W_1, \ldots, W_n\}$. Then,
\[
\sigma_{XY \cdot ZW_1} =  \sigma_{XY \cdot Z}
\]
by replacing $W$ with $W_1$ in the paragraph above. Likewise, 
\[
\sigma_{XY \cdot ZW_1W_2} =  \sigma_{XY \cdot ZW_1}
\]
by replacing $Z$ and $W$ with $Z \cup \{W_1\}$ and $W_2$, respectively, in the paragraph above. These last two results imply that 
\[
\sigma_{XY \cdot ZW_1W_2} =  \sigma_{XY \cdot Z}.
\]
Continuing with this process for $W_3, \ldots, W_n$ yields the desired result.
\end{proof}

\begin{lemma}\label{lem:nonrootch}
Let $S$ be a non-root node in a path $\pi_{XY}$ without colliders, i.e. $A \oa S \ra B$ or $A \oa S$ is a subpath of $\pi_{XY}$. Note that $S=X$ or $S=Y$ in the latter case. Let $W$ be a set of nodes such that each is connected to $Ch(S)$ by a path that does not contain any node in $\pi_{XY}$. Then,
\[
\sigma_{XY \cdot ZW} =  \sigma_{XY \cdot Z} \frac{\sigma^2_{S \cdot ZW}}{\sigma^2_{S \cdot Z}}.
\]
\end{lemma}

\begin{proof}
Assume that $W$ is a singleton. Consider first the case where $A \oa S \ra B$ is a subpath of $\pi_{XY}$. Note that $X \ci W | Z \cup S$. Then,
\[
0 = \sigma_{XW \cdot ZS} = \sigma_{XW \cdot Z} - \frac{\sigma_{XS \cdot Z} \sigma_{SW \cdot Z}}{\sigma^2_{S \cdot Z}}
\]
which implies that $\sigma_{XW \cdot Z} = \sigma_{XS \cdot Z} \delta_{SW \cdot Z}$ where $\delta_{SW \cdot Z} = \sigma_{SW \cdot Z} / \sigma^2_{S \cdot Z}$. Note also that $Y \ci W | Z \cup S$. Then,
\[
0 = \sigma_{YW \cdot ZS} = \sigma_{YW \cdot Z} - \frac{\sigma_{YS \cdot Z} \sigma_{SW \cdot Z}}{\sigma^2_{S \cdot Z}}
\]
which implies that $\sigma_{YW \cdot Z} = \delta_{YS \cdot Z} \sigma_{SW \cdot Z}$ where $\delta_{YS \cdot Z} = \sigma_{YS \cdot Z} / \sigma^2_{S \cdot Z}$. Likewise, $X \ci Y | Z \cup S$ implies that
\[
0 = \sigma_{XY \cdot ZS} = \sigma_{XY \cdot Z} - \frac{\sigma_{XS \cdot Z} \sigma_{SY \cdot Z}}{\sigma^2_{S \cdot Z}}
\]
which implies that $\sigma_{XY \cdot Z} = \sigma_{XS \cdot Z} \delta_{YS \cdot Z}$. Therefore,
\begin{align*}
\sigma_{XY \cdot ZW} &= \sigma_{XY \cdot Z} - \frac{\sigma_{XW \cdot Z} \sigma_{WY \cdot Z}}{\sigma^2_{W \cdot Z}} = \sigma_{XY \cdot Z} - \frac{\sigma_{XS \cdot Z} \delta_{SW \cdot Z} \delta_{YS \cdot Z} \sigma_{SW \cdot Z}}{\sigma^2_{W \cdot Z}}\\
&= \sigma_{XY \cdot Z} \Big(1 - \frac{\delta_{SW \cdot Z} \sigma_{SW \cdot Z}}{\sigma^2_{W \cdot Z}} \Big) = \frac{\sigma_{XY \cdot Z}}{\sigma^2_{S \cdot Z}} \Big(\sigma^2_{S \cdot Z} - \frac{\sigma^2_{S \cdot Z} \delta_{SW \cdot Z} \sigma_{SW \cdot Z}}{\sigma^2_{W \cdot Z}} \Big)\\
& = \sigma_{XY \cdot Z} \frac{\sigma^2_{S \cdot ZW}}{\sigma^2_{S \cdot Z}}.
\end{align*}

Now, consider the case where $A \oa S$ is a subpath of $\pi_{XY}$. Assume without loss of generality that $S=Y$. Note that $X \ci W | Z \cup Y$. Then,
\[
0 = \sigma_{XW \cdot ZY} = \sigma_{XW \cdot Z} - \frac{\sigma_{XY \cdot Z} \sigma_{YW \cdot Z}}{\sigma^2_{Y \cdot Z}}
\]
which implies that
\[
\sigma_{XW \cdot Z} = \frac{\sigma_{XY \cdot Z} \sigma_{YW \cdot Z}}{\sigma^2_{Y \cdot Z}}.
\]
Therefore,
\begin{align*}
\sigma_{XY \cdot ZW} &= \sigma_{XY \cdot Z} - \frac{\sigma_{XW \cdot Z} \sigma_{WY \cdot Z}}{\sigma^2_{W \cdot Z}} = \sigma_{XY \cdot Z} - \frac{\sigma_{XY \cdot Z} \sigma_{YW \cdot Z} \sigma_{WY \cdot Z}}{\sigma^2_{Y \cdot Z} \sigma^2_{W \cdot Z}}\\
&= \sigma_{XY \cdot Z} \Big( 1 - \frac{\sigma_{YW \cdot Z} \sigma_{WY \cdot Z}}{\sigma^2_{Y \cdot Z} \sigma^2_{W \cdot Z}} \Big) = \frac{\sigma_{XY \cdot Z}}{\sigma^2_{Y \cdot Z}} \Big( \sigma^2_{Y \cdot Z} - \frac{\sigma_{YW \cdot Z} \sigma_{WY \cdot Z}}{\sigma^2_{W \cdot Z}} \Big)\\
& = \sigma_{XY \cdot Z} \frac{\sigma^2_{Y \cdot ZW}}{\sigma^2_{Y \cdot Z}}.
\end{align*}

Repeated application of the paragraphs above proves the result for when $W$ is a set. Specifically, let $W=\{W_1, \ldots, W_n\}$. Then,
\[
\sigma_{XY \cdot ZW_1} =  \sigma_{XY \cdot Z} \frac{\sigma^2_{S \cdot ZW_1}}{\sigma^2_{S \cdot Z}}
\]
by replacing $W$ with $W_1$ in the paragraphs above. Likewise, 
\[
\sigma_{XY \cdot ZW_1W_2} =  \sigma_{XY \cdot ZW_1} \frac{\sigma^2_{S \cdot ZW_1W_2}}{\sigma^2_{S \cdot ZW_1}}
\]
by replacing $Z$ and $W$ with $Z \cup \{W_1\}$ and $W_2$, respectively, in the paragraphs above. These last two results imply that 
\[
\sigma_{XY \cdot ZW_1W_2} =  \sigma_{XY \cdot Z} \frac{\sigma^2_{S \cdot ZW_1W_2}}{\sigma^2_{S \cdot Z}}.
\]
Continuing with this process for $W_3, \ldots, W_n$ yields the desired result.
\end{proof}

\begin{proof}[Proof of Theorem \ref{the:condpath1}]
First, note that
\[
\sigma_{X_m X_{m+n} \cdot Z_1^1} = \sigma_{X_m X_{m+n}} \frac{\sigma^2_{X_1 \cdot Z_1^1}}{\sigma^2_{X_1}}
\]
by Lemma \ref{lem:root}. Then, note that
\[
\sigma_{X_m X_{m+n} \cdot Z_1^1 Z^2} = \sigma_{X_m X_{m+n} \cdot Z_1^1}
\]
by Lemma \ref{lem:nonrootpasp}. Finally, note that
\[
\sigma_{X_m X_{m+n} \cdot Z_1^1 Z^2 Z_2} = \sigma_{X_m X_{m+n} \cdot Z_1^1 Z^2} \frac{\sigma^2_{X_1 \cdot Z_1^1 Z^2 Z_2}}{\sigma^2_{X_1 \cdot Z_1^1 Z^2}}
\]
by Lemma \ref{lem:nonrootch}. Continuing with this process for the rest of the nodes yields the desired result.
\end{proof}

\begin{proof}[Proof of Theorem \ref{the:condpath2}]
First, note that
\[
\sigma_{X_m X_{m+n} \cdot Z^1} = \sigma_{X_m X_{m+n}}
\]
by Lemma \ref{lem:nonrootpasp}. Then, note that
\[
\sigma_{X_m X_{m+n} \cdot Z^1 Z_1} = \sigma_{X_m X_{m+n} \cdot Z^1} \frac{\sigma^2_{X_1 \cdot Z^1 Z_1}}{\sigma^2_{X_1 \cdot Z^1}}
\]
by Lemma \ref{lem:nonrootch}. Continuing with this process for the rest of the nodes yields the desired result.
\end{proof}

\begin{proof}[Proof of Theorem \ref{the:collider}]
First, note that
\begin{align*}
\sigma_{XY \cdot Z Z_{1:n}^{1:n} W_{1:n}} &= \sigma_{XY \cdot Z Z_{1:n-1}^{1:n} W_{1:n}}\\
&= \sigma_{XY \cdot Z Z_{1:n-1}^{1:n} W_{1:n-1}} - \frac{\sigma_{X W_n \cdot Z Z_{1:n-1}^{1:n} W_{1:n-1}} \sigma_{W_n Y \cdot Z Z_{1:n-1}^{1:n} W_{1:n-1}}}{\sigma^2_{W_n \cdot Z Z_{1:n-1}^{1:n} W_{1:n-1}}}\\
&= \sigma_{XY \cdot Z Z_{1:n-1}^{1:n-1} W_{1:n-1}} - \frac{\sigma_{X W_n \cdot Z Z_{1:n-1}^{1:n} W_{1:n-1}} \sigma_{W_n Y \cdot Z Z_{1:n-1}^{1:n} W_{1:n-1}}}{\sigma^2_{W_n \cdot Z Z_{1:n-1}^{1:n} W_{1:n-1}}}
\end{align*}
because $X \cup Y \ci Z_n | Z \cup Z_{1:n-1}^{1:n} \cup W_{1:n}$ and $X \cup Y \ci Z^n | Z \cup Z_{1:n-1}^{1:n-1} \cup W_{1:n-1}$. Then, the theorem follows by recursively applying the paragraph above to $\sigma_{XY \cdot Z Z_{1:n-1}^{1:n-1} W_{1:n-1}}$ until $n-1=0$, in which case $\sigma_{XY \cdot Z Z_{1:n-1}^{1:n-1} W_{1:n-1}} = \sigma_{XY \cdot Z} = 0$ because $X \ci Y | Z$.
\end{proof}

\begin{lemma}\label{lem:Simpson}
Let $\pi_{XY}$ be a path that is open with respect to a set of nodes $U$. Then, the sign of $\sigma_{XY \cdot U}$ does not depend on $U$, i.e. $sign(\sigma_{XY})=sign(\sigma_{XY \cdot U})$.
\end{lemma}

\begin{proof}
If $\pi_{XY}$ has no colliders, then the result follows from Corollary \ref{cor:samesign}. Otherwise, consider any collider $C$ in $\pi_{XY}$, and split $U$ into $Z \cup Z_{1:n}^{1:n} \cup W_{1:n}$ as indicated in Theorem \ref{the:collider}. We prove the result by induction over the number of colliders in $\pi_{XY}$. First, assume that $C$ is the only collider in $\pi_{XY}$. Consider the following two cases.
\begin{itemize}
\item[Case 1.] Suppose that $C$ is an opener for the collider $C$. Then, $C$ is the only opener by definition and, thus, the summation in Equation \ref{eq:collider} reduces to
\[
\frac{\sigma_{X W_1 \cdot Z Z^{1}} \sigma_{W_1 Y \cdot Z Z^{1}}}{\sigma^2_{W_1 \cdot Z Z^{1}}}
\]
with $W_1=C$.

\item[Case 2.] Suppose that $C$ is not an opener for the collider $C$. Then, note that $X \ci W_i | Z \cup Z_{1:i-1}^{1:i} \cup W_{1:i-1} \cup C$ for all $i$. Then,
\[
0 = \sigma_{X W_i \cdot Z Z_{1:i-1}^{1:i} W_{1:i-1} C} = \sigma_{X W_i \cdot Z Z_{1:i-1}^{1:i} W_{1:i-1}} - \frac{\sigma_{X C \cdot Z Z_{1:i-1}^{1:i} W_{1:i-1}} \sigma_{C W_i \cdot Z Z_{1:i-1}^{1:i} W_{1:i-1}}}{\sigma^2_{C \cdot Z Z_{1:i-1}^{1:i} W_{1:i-1}}}
\]
and thus
\[
\sigma_{X W_i \cdot Z Z_{1:i-1}^{1:i} W_{1:i-1}} = \frac{\sigma_{X C \cdot Z Z_{1:i-1}^{1:i} W_{1:i-1}} \sigma_{C W_i \cdot Z Z_{1:i-1}^{1:i} W_{1:i-1}}}{\sigma^2_{C \cdot Z Z_{1:i-1}^{1:i} W_{1:i-1}}}.
\]
Likewise for $\sigma_{W_i Y \cdot Z Z_{1:i-1}^{1:i} W_{1:i-1}}$. Then, each term in the summation in Equation \ref{eq:collider} can be rewritten as
\[
\frac{\sigma_{X C \cdot Z Z_{1:i-1}^{1:i} W_{1:i-1}} \sigma_{C Y \cdot Z Z_{1:i-1}^{1:i} W_{1:i-1}} \sigma^2_{C W_i \cdot Z Z_{1:i-1}^{1:i} W_{1:i-1}}}{\sigma^2_{W_i \cdot Z Z_{1:i-1}^{1:i} W_{1:i-1}} \sigma^4_{C \cdot Z Z_{1:i-1}^{1:i} W_{1:i-1}}}.
\]
\end{itemize}
In both cases above, $\pi_{X C}$ and $\pi_{Y C}$ have no colliders and, thus, the signs of $\sigma_{X C \cdot M}$ and $\sigma_{C Y \cdot N}$ do not depend on $M$ and $N$ by Corollary \ref{cor:samesign} for any sets of nodes $M$ and $N$. Therefore, the sign of $\sigma_{XY \cdot Z Z_{1:n}^{1:n} W_{1:n}}$ is the same in both cases above and, moreover, it does not depend on $Z \cup Z_{1:n}^{1:n} \cup W_{1:n}$.

After proving above that the result holds for any path with zero or one collider, we now assume as induction hypothesis that the result holds for any path with fewer than $k$ colliders. To prove it for $k$ colliders, we simply let $C$ be any collider in $\pi_{XY}$ and consider the same two cases as above. Note that $\pi_{X C}$ and $\pi_{C Y}$ may now have colliders. So, Corollary \ref{cor:samesign} does not apply. However, since $\pi_{X C}$ and $\pi_{C Y}$ have fewer than $k$ colliders, the induction hypothesis does apply, which leads to the same conclusions as before.
\end{proof}

\begin{proof}[Proof of Theorem \ref{the:Simpson}]
It follows from Lemma \ref{lem:Simpson}.
\end{proof}

\input{AppendixB.tex}
\input{AppendixD.tex}
\input{AppendixE.tex}
\input{AppendixF.tex}

\bibliographystyle{plainnat}
\bibliography{CondPathAnalysis}

\end{document}

%% file: AppendixB.tex
\section*{Appendix B: Suboptimal Decision Making I}

In this appendix, we show that the bias introduced by conditioning on a child of the effect (recall Example \ref{exa:fig4and5}) may lead to suboptimal decision making. We do so with the help of the following fictitious but, in our opinion, realistic scenario. Doctor 1 and Doctor 2 both treat a certain disease by administering approximately 5 units of drug $X$, i.e. $X \sim \mathcal{N}(5,\sigma_X)$. The doctors use different methods to administer the drug, which we suspect affects the effectiveness of the drug. The effectiveness of the drug is assessed by measuring the abundance of Y in blood, which is determined by $X$, i.e. $Y = \alpha_i X + \epsilon_Y$ for Doctor $i$ and $\epsilon_Y \sim \mathcal{N}(0,\sigma_Y)$. The higher the value of $Y$ the higher the effectiveness of the treatment. Moreover, the doctors also monitor the abundances of $Z$ and $W$ in blood, which are determined by respectively $X$ and $Y$, specifically $Z = X + \epsilon_Z$ and $ W = Y + \epsilon_W$ for both doctors and $\epsilon_Z \sim \mathcal{N}(0,\sigma_Z)$ and $\epsilon_W \sim \mathcal{N}(0,\sigma_W)$. The doctors divide the treatments into ordinary and extraordinary. Specifically, Doctor 1 declares the treatment ordinary when $4 < Z < 6$, and Doctor 2 when $4 < W < 6$. The doctors share with us data only about ordinary treatments. They believe that extraordinary treatments may lead to new findings about the disease at hand and, thus, they are not willing to share them as of today.

The problem above can be rephrased as follows. We want to estimate $\alpha_1$ in the following path diagram (Doctor 1) from a sample of the subpopulation satisfying $4 < Z < 6$:
\begin{center}
\begin{tikzpicture}[inner sep=1mm]
\node at (0,0) (X) {$X$};
\node at (0,-1.2) (Z) {$Z$};
\node at (1.5,0) (Y) {$Y$};
\path[->] (X) edge node[above] {$\alpha_1$} (Y);
\path[->] (X) edge node[right] {$1$} (Z);
\end{tikzpicture}
\end{center}
We also want to estimate $\alpha_2$ in the following path diagram (Doctor 2) from a sample of the subpopulation satisfying $4 < W < 6$:
\begin{center}
\begin{tikzpicture}[inner sep=1mm]
\node at (0,0) (X) {$X$};
\node at (1.5,-1.2) (Z) {$W$};
\node at (1.5,0) (Y) {$Y$};
\path[->] (X) edge node[above] {$\alpha_2$} (Y);
\path[->] (Y) edge node[right] {$1$} (Z);
\end{tikzpicture}
\end{center}
As discussed in Example \ref{exa:fig4and5}, the estimate of $\alpha_1$ will be unbiased, whereas the estimate of $\alpha_2$ will be biased. This may make us recommend the suboptimal doctor to future patients. We illustrate this below with some experiments.

\begin{figure}
\includegraphics[scale=.4]{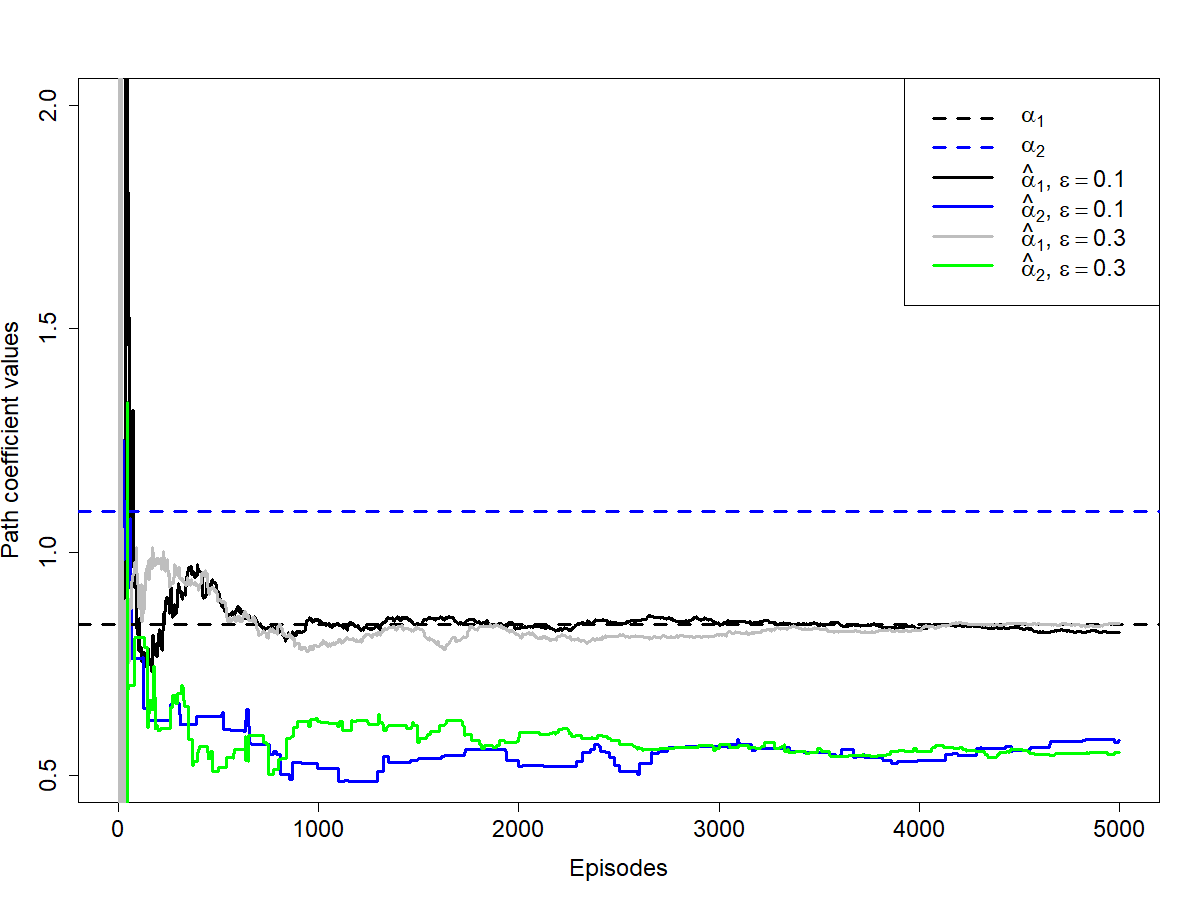}
\includegraphics[scale=.4]{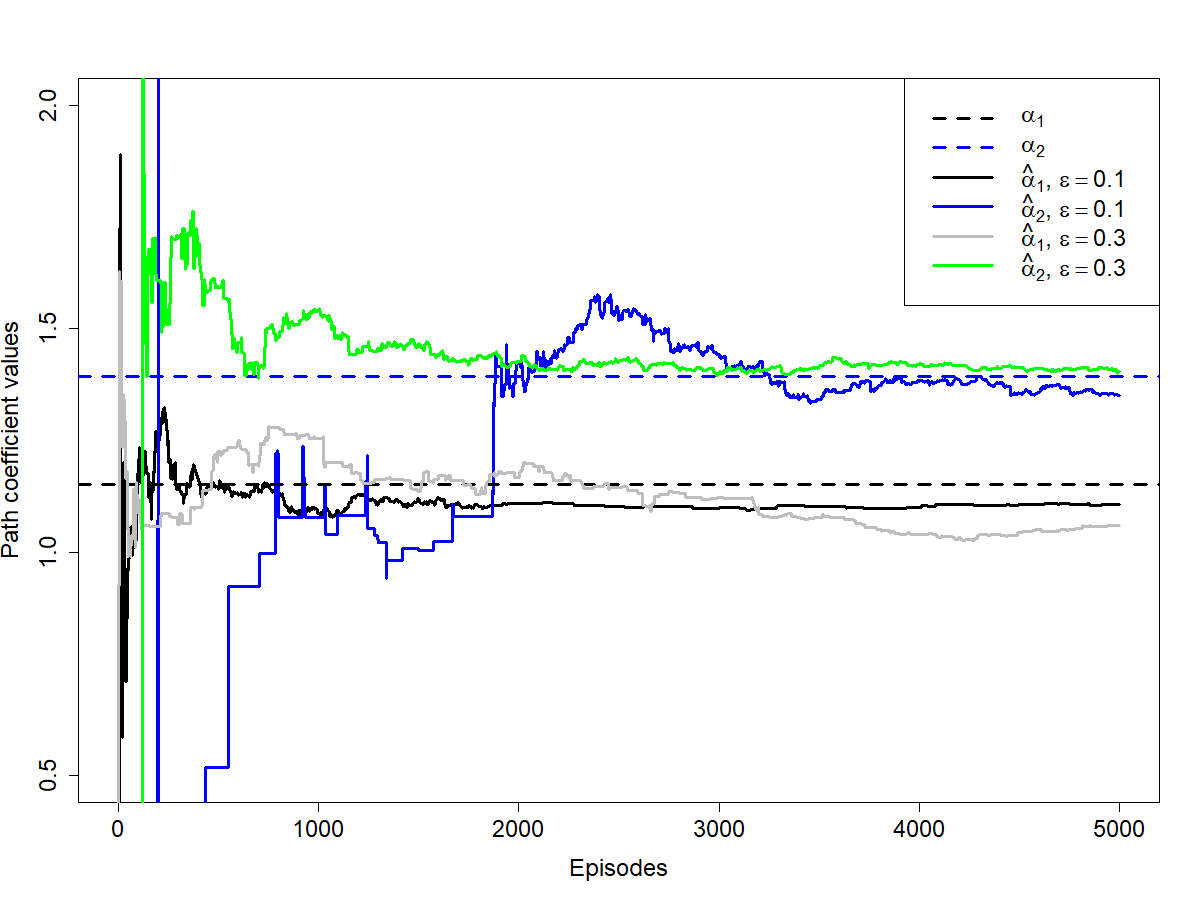}
\caption{Estimates of $\alpha_1$ and $\alpha_2$ obtained by $\epsilon$-greedy without correction (top) and with correction (bottom).}\label{fig:rl}
\end{figure}

Estimating $\alpha_1$ and $\alpha_2$ above can be seen as an instance of the exploration/exploitation dilemma: In order to learn the effectiveness of the treatment administered by a particular doctor, the doctor has to administer the treatment to some patients, which leads to some patients receiving suboptimal treatment. A straightforward solution to this dilemma consists in combining exploration and exploitation as follows: Select the doctor with the highest effectiveness so far (i.e. exploitation) with probability $1-\epsilon$, otherwise select the doctor at random (i.e. exploration). This strategy is called $\epsilon$-greedy in the reinforcement learning community \citep{SuttonandBarto2018}.

Figure \ref{fig:rl} (top) shows the estimates of $\alpha_1$ and $\alpha_2$ (denoted as $\hat{\alpha}_1$ and $\hat{\alpha}_2$) obtained by running $\epsilon$-greedy for 5000 iterations (a.k.a. episodes). Each episode consists in selecting a doctor for treating a patient. The doctor shares the data with us only if the treatment is regarded as ordinary. In our experiments, this means that each episode starts by choosing a doctor, say Doctor 1, according to the $\epsilon$-greedy strategy. Then, a triplet of values $(x,y,z)$ is sampled from the corresponding linear structural equation model.\footnote{Since negative abundance values do not make sense, if $x$, $y$ or $z$ are negative then a new triplet is sampled. This seldom happens, anyway.} Finally, the triplet is kept if $4<z<6$ and discarded otherwise. In the figure, we can clearly see that $\hat{\alpha}_1$ converges to $\alpha_1$, whereas $\hat{\alpha}_2$ does not converge to $\alpha_2$. Moreover, $\hat{\alpha}_1$ converges to a larger value than $\hat{\alpha}_2$, which means that Doctor 1 is considered more effective than Doctor 2 and, thus, we should recommend the former. This is suboptimal because, as shown in the figure, $\alpha_2$ is greater than $\alpha_1$ and, thus, Doctor 2 should be preferred. This conclusion was consistent across many runs of the experiment. In each run, $\alpha_1$ and $\alpha_2$ were sampled uniformly from the intervals $(0.5, 1.5)$ and $(\alpha_1 + 0.15, \alpha_1 + 0.3)$ respectively, i.e. Doctor 2 was more effective than Doctor 1. In each run, $\sigma_X = \sigma_Y = \sigma_Z = \sigma_W = 1$.

As discussed in Example \ref{exa:fig4and5}, if we can estimate $\sigma_X^2$ and $\sigma_Y^2$, then we can correct the bias in $\hat{\alpha}_2$. To illustrate this, assume that the doctors do not share with us data about individual extraordinary treatments but they do share aggregated data, in particular some estimates of $\sigma_X^2$ and $\sigma_Y^2$ (which they can compute from all the ordinary and extraordinary treatments performed). Figure \ref{fig:rl} (bottom) shows $\hat{\alpha}_1$ and $\hat{\alpha}_2$ when the correction is applied to the latter. We can appreciate that both path coefficient estimates converge to the true values, and that Doctor 2 is now preferred. Again, this conclusion was consistent across many runs of the experiment.

Of course, $\epsilon$-greedy is not the only way of solving the problem above. Alternative solutions include Thompson sampling, upper confidence bound (UCB) or directly performing a randomized controlled trial. However, the conclusions should not differ essentially from the ones presented above. The code for our experiments is publicly available at \texttt{https://www.dropbox.com/s/hawshrihhgr5uvi/MAB.zip?dl=0}.

%% file: AppendixD.tex
\section*{Appendix C: Suboptimal Decision Making II}

In this appendix, we show that the bias introduced by adjusting for a faithful proxy of a confounder is negligible for decision making (recall Example \ref{exa:fig13}). However, the bias may be substantial when adjusting for a proxy of a non-confounder in a confounding path, which may lead to suboptimal decision making. We do so with the help of the following fictitious but, in our opinion, realistic scenario. Doctor 1 and Doctor 2 both treat a certain disease by administering a dose of drug $X$. The dose is determined by the abundance of $U$ in blood. The doctors use different methods to administer the drug, which we suspect affects the effectiveness of the drug. The effectiveness of the drug is assessed by measuring the abundance of Y in blood, which is determined by $X$ and $U$. The lower the value of $Y$ the higher the effectiveness of the treatment. Unwilling to disclose further details about the treatment, the doctors do not share with us any measurements of $U$. However, they do provide us with measurements of two proxies of $U$. Specifically, Doctor 1 provides us with the abundance of $Z$ in blood, whereas Doctor 2 provides us with the abundance of $W$ in blood. The former is known to be caused by $U$, whereas the latter is known to cause $U$.\footnote{Some authors would call $Z$ a proxy and $W$ a driver.}

In the language of path diagrams, the problem above can be stated as follows. We want to estimate $\alpha_1$ in the following path diagram (Doctor 1) from a sample for $X$, $Y$ and $Z$:
\begin{center}
\begin{tikzpicture}[inner sep=1mm]
\node at (0,0) (X) {$X$};
\node at (2.5,1.2) (Z) {$Z$};
\node at (2,0) (Y) {$Y$};
\node at (1,1.2) (U) {$U$};
\path[->] (X) edge node[above] {$\alpha_1$} (Y);
\path[->] (U) edge node[left] {$1$} (X);
\path[->] (U) edge node[right] {$1$} (Y);
\path[->] (U) edge node[above] {$1$} (Z);
\end{tikzpicture}
\end{center}
We also want to estimate $\alpha_2$ in the following diagram (Doctor 2) from a sample for $X$, $Y$ and $W$:
\begin{center}
\begin{tikzpicture}[inner sep=1mm]
\node at (0,0) (X) {$X$};
\node at (2.5,1.2) (Z) {$W$};
\node at (2,0) (Y) {$Y$};
\node at (1,1.2) (U) {$U$};
\path[->] (X) edge node[above] {$\alpha_2$} (Y);
\path[->] (U) edge node[left] {$1$} (X);
\path[->] (U) edge node[right] {$1$} (Y);
\path[<-] (U) edge node[above] {$1$} (Z);
\end{tikzpicture}
\end{center}
Recall that $U$ is unobserved. As discussed in Example \ref{exa:fig13}, if $Z$ and $W$ are faithful proxies of $U$, then the estimates of $r_{YX \cdot Z}$ and $r_{YX \cdot W}$ should be close to $\alpha_1$ and $\alpha_2$, respectively, which implies that we should be able to identify the optimal doctor. We illustrate this below with some experiments.

\begin{figure}
\includegraphics[scale=.4]{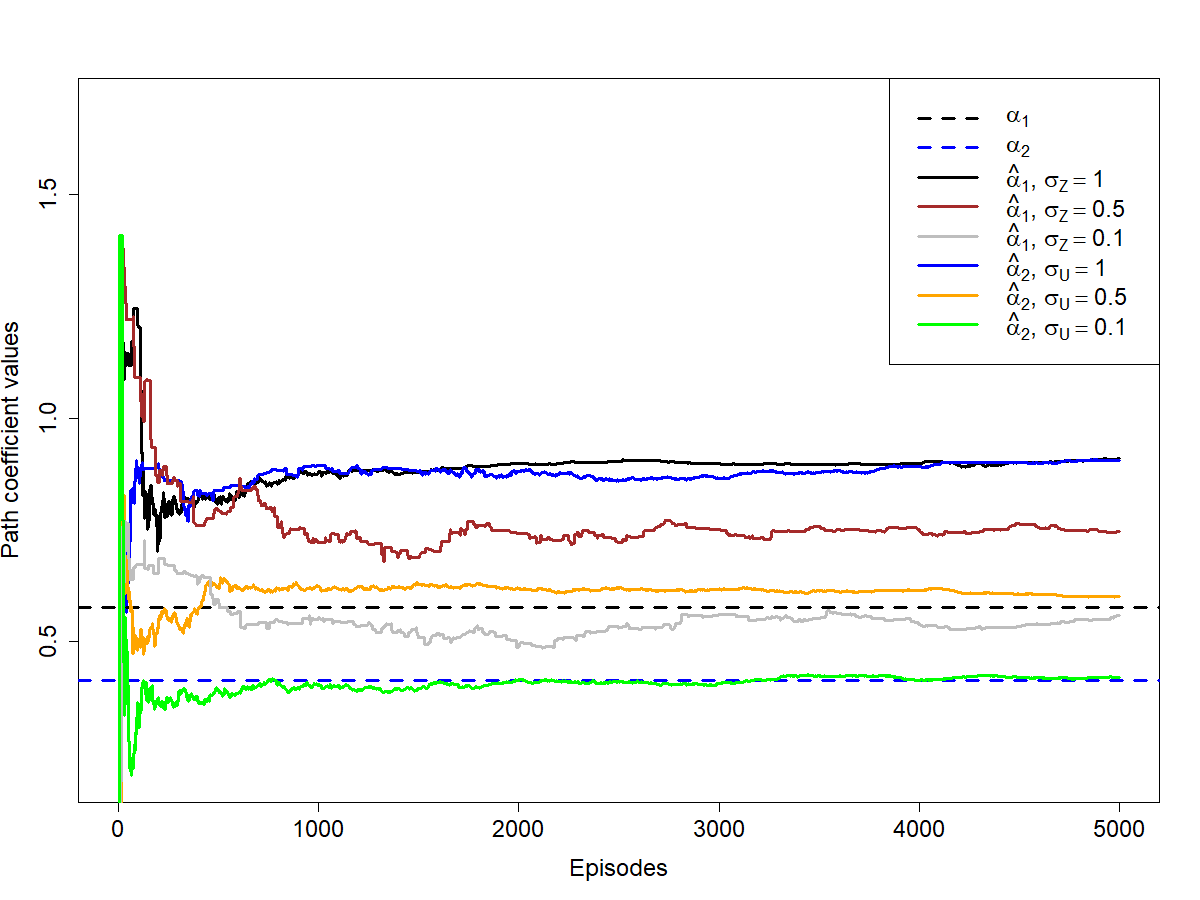}
\includegraphics[scale=.4]{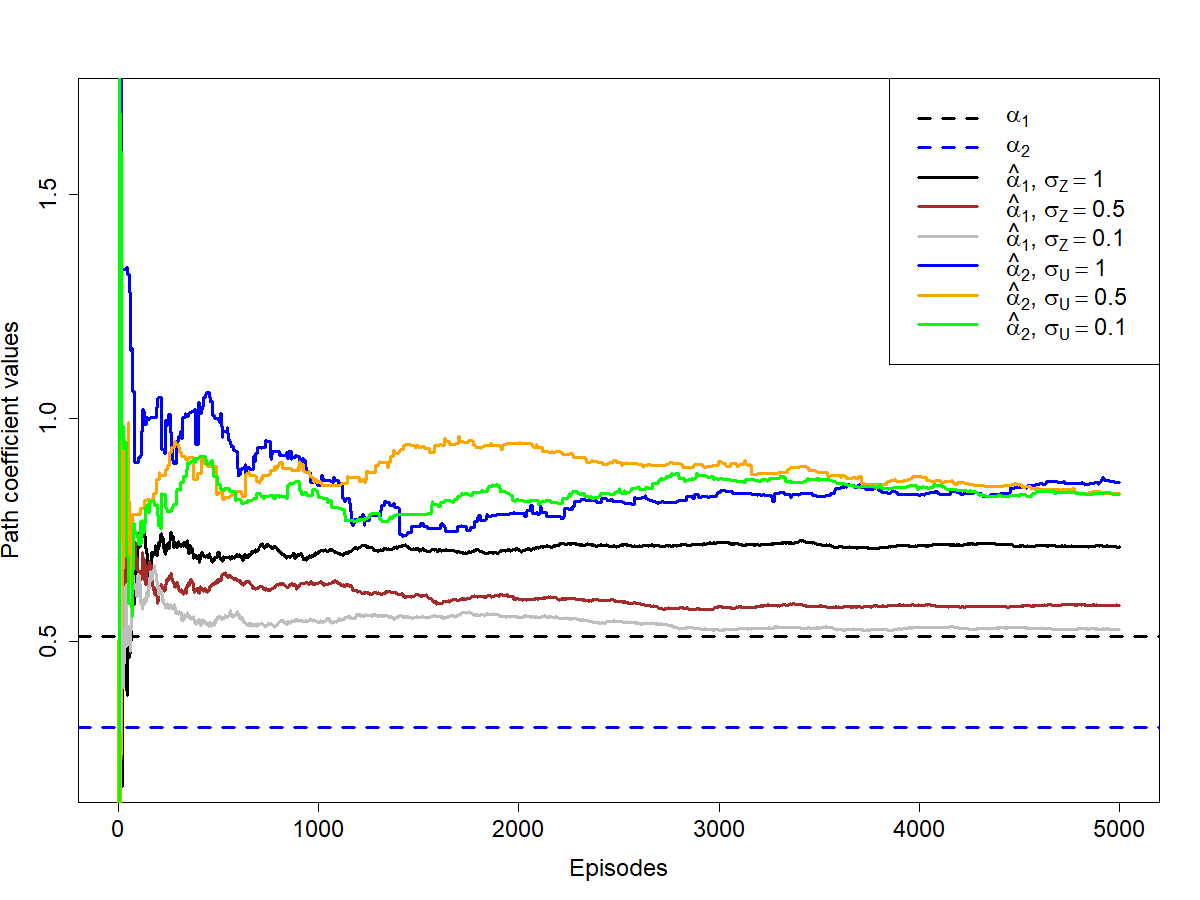}
\caption{Estimates of $\alpha_1$ and $\alpha_2$ obtained by $\epsilon$-greedy when adjusting for a proxy of a confounder (top) and for a proxy of a non-confounder (bottom).}\label{fig:rl2}
\end{figure}
 
As in Appendix B, we use $\epsilon$-greedy to solve the problem above. We consider $\epsilon=0.2$ and 5000 episodes. In each run of $\epsilon$-greedy, $\alpha_1$ and $\alpha_2$ are sampled uniformly from the intervals $(0.5, 1.5)$ and $(\alpha_1 - 0.3, \alpha_1 - 0.15)$ respectively, i.e. Doctor 2 is more effective than Doctor 1. The standard deviations of the error terms are all equal to 1, with the exception of the term corresponding to $Z$ for Doctor 1 and the term corresponding to $U$ for Doctor 2. Specifically, $\sigma_Z, \sigma_{U} = 0.1, 0.5, 1$. The smaller the values of $\sigma_Z$ and $\sigma_{U}$ the better $Z$ and $W$ are as proxies of $U$. Figure \ref{fig:rl2} (top) shows a representative run of the many that we performed. We can see that Doctor 2 is preferred if $Z$ and $W$ are equally good proxies of $U$, i.e. $\sigma_Z=\sigma_{U}$. Moreover, both $\hat{\alpha}_1$ and $\hat{\alpha}_2$ converge to the true values when $Z$ and $W$ are faithful proxies of $U$, i.e. $\sigma_Z=\sigma_{U}=0.1$.

The experiments above may lead one to conclude that blocking a confounding path by adjusting for a proxy does not bias much the estimate of a causal effect as long as the proxy is a good one. However, this is not true. To illustrate it, we repeat the experiments above after replacing the confounding path from $X$ to $Y$ in the path diagrams with the confounding path $X \la U' \ra U \ra Y$. Additional parameters are all set to 1. Figure \ref{fig:rl2} (bottom) shows a representative run of the new experiments. We can clearly see that $\hat{\alpha}_1$ converges to a smaller value than $\hat{\alpha}_2$ for every combination of $\sigma_Z$ and $\sigma_{U}$, i.e. no matter how good $Z$ and $W$ are as proxies of $U$. This means that Doctor 1 is considered more effective than Doctor 2 and, thus, that we should recommend the former. This is suboptimal because, as shown in the figure, $\alpha_2$ is smaller than $\alpha_1$ and, thus, Doctor 2 should be preferred. Note also that $\hat{\alpha}_1$ converges to $\alpha_1$ when $Z$ is almost a perfect proxy of $U$. On the other hand, $\hat{\alpha}_2$ behaves bad no matter how good $W$ is as a proxy of $U$. In summary, on the negative side, we wrongly recommend Doctor 1 but, on the positive side, we can estimate her effectiveness accurately if $Z$ is a faithful proxy of $U$. To get further insight into these results, we can repeat the reasoning in Example \ref{exa:fig13} now for the path diagrams of Doctor 1 and Doctor 2 after replacing the confounding path from $X$ to $Y$ with the confounding path $X \la U' \ra U \ra Y$. Let $G$ denote the path diagram of Doctor 1. Since $G^{\alpha_1}$ is singly-connected, we can apply Theorem \ref{the:condpath1} and conclude that $\sigma_{XY \cdot Z}^{\alpha_1} = \sigma_{XY}^{\alpha_1} \sigma_{U \cdot Z}^2 / \sigma_{U}^2$. This implies that conditioning on $Z$ reduces the covariance between $X$ and $Y$. Moreover, the greater the correlation between $U$ and $Z$, the greater the reduction and, thus, the closer $r_{YX \cdot Z}$ comes to $\alpha_1$. Let $G$ now denote the path diagram of Doctor 2. Applying Theorem \ref{the:condpath1} to $G^{\alpha_2}$ gives that $\sigma_{XY \cdot W}^{\alpha_2} = \sigma_{XY}^{\alpha_2}$. In other words, conditioning on $W$ leaves the covariance of $X$ and $Y$ unchanged. Moreover, $X \ci W | \emptyset$ in $G$ and, thus, $\sigma_{X \cdot W} = \sigma_X$ and, thus, $r_{YX \cdot W} = r_{YX}$. In other words, adjusting for $W$ does not solve our problem, even if $W$ is almost a perfect proxy of $U$. In summary, the effectiveness of adjusting for a proxy depends on the type of confounding path, the type of causal relation between the proxy and the unobserved variable, and the correlation between them.

%% file: AppendixE.tex
\section*{Appendix D: Beyond Singly-Connected Path Diagrams}

In this appendix, we extend Theorems \ref{the:condpath1} and \ref{the:condpath2} from singly-connected diagrams to a superclass thereof. Note that we then only consider paths without colliders. The extension to path with colliders seems complicated.

In Section \ref{sec:nocolliders}, we defined the separation criterion for path diagrams in terms of paths. For some of the results in this appendix, it is more convenient to define it in terms of routes. Recall that whereas all the nodes in a path must be different, the nodes in a route do not need to be so. Given a route $\rho_{X:Y}$ from a node $X$ to a node $Y$ in a path diagram, a node $C$ is a collider in $\rho_{X:Y}$ if $A \oa C \ao B$ is a subroute of $\rho_{X:Y}$. Note that $A$ and $B$ may be the same node. Given a set of nodes $Z$, $\rho_{X:Y}$ is said to be $Z$-open if
\begin{itemize}
\item every collider in $\rho_{X:Y}$ is in $Z$, and
\item ever non-collider in $\rho_{X:Y}$ is outside $Z$.
\end{itemize}
Note that there is a $Z$-open route from $X$ to $Y$ if and only if there is a $Z$-open path from $X$ to $Y$ (see Lemma \ref{lem:pathroute} in Appendix E). When such a path or route exists, we say that $X$ and $Y$ are $Z$-connected.

Before presenting the results in this appendix, we define the operation of conditioning a path diagram on a node $A$ as replacing every edge $A \ra B$ with an edge $A_B \ra B$, where $A_B$ is a new node. Note that $A$ is not removed. In terms of the associated system of linear equations, this implies (i) adding a new equation $A_B = \epsilon_{A_B}$ where $\epsilon_{A_B}$ is normally distributed with arbitrary mean and variance, and (ii) replacing every equation $B = \alpha^T (A, Pa(B) \setminus A) + \epsilon_B$ with an equation $B = \alpha^T (A_B, Pa(B) \setminus A) + \epsilon_B$. Note that, after conditioning, we have that $Ch(A)=\emptyset$ whereas $Pa(A_B) \cup Sp(A_B)=\emptyset$ and $Ch(A_B)=B$. See Figure \ref{fig:illustration} for an illustration. Let $V$ denote all the nodes in the path diagram at hand, and consider the distribution $p(V \setminus A, A=a)$ defined by the system of equations before conditioning on $A$. This is the unnormalized conditional distribution of $V \setminus A$ given $A$. Let $A'$ denote the new nodes created by conditioning on $A$, and consider the distribution $p(V \setminus A, A=a, A'=a)$ defined by the system of equations after conditioning on $A$. This is the unnormalized conditional distribution of $V \setminus A$ given $A \cup A'$. Note that both unnormalized conditional distributions coincide for all $a$, i.e. $p(V \setminus A = x, A=a) = p(V \setminus A = x, A=a, A'=a)$ for all $x$ and $a$. Thus, their normalized versions coincide as well. So, computing partial covariances in either of them gives the same result, since partial covariances coincide with conditional covariances for Gaussian random vectors. In other words, the partial covariance $\sigma_{XY \cdot A}$ in the original path diagram is equal to $\sigma_{XY \cdot A A'}$ in the conditional diagram. Finally, we define conditioning on a set of nodes $S$ as conditioning on each node in $S$. By the previous reasoning, the partial covariance $\sigma_{XY \cdot S}$ in the original path diagram is equal to $\sigma_{XY \cdot S S'}$ in the conditional diagram, where $S'$ denotes the new nodes created by conditioning on $S$. The following theorems show how to compute the latter. See the Appendix E for the proofs. We illustrate the theorems through some examples afterwards.

\begin{figure}
\begin{tabular}{c|c}
\begin{tabular}{c}
\begin{tikzpicture}[inner sep=1mm]
\node at (0,0) (W) {$A$};
\node at (-1,1) (X1) {};
\node at (1,1) (X2) {};
\node at (-1.3,0) (X3) {};
\node at (1.3,0) (X4) {};
\node at (-1,-1) (X5) {$B$};
\node at (1,-1) (X6) {$C$};
\path[->] (X1) edge (W);
\path[->] (X2) edge (W);
\path[<->] (X3) edge (W);
\path[<->] (X4) edge (W);
\path[<-] (X5) edge (W);
\path[<-] (X6) edge (W);
\end{tikzpicture}
\end{tabular}
&
\begin{tabular}{c}
\begin{tikzpicture}[inner sep=1mm]
\node at (0,0) (W) {$A$};
\node at (-0.5,-1) (W1) {$A_B$};
\node at (0.5,-1) (W2) {$A_C$};
\node at (-1,1) (X1) {};
\node at (1,1) (X2) {};
\node at (-1.3,0) (X3) {};
\node at (1.3,0) (X4) {};
\node at (-1.5,-2) (X5) {$B$};
\node at (1.5,-2) (X6) {$C$};
\path[->] (X1) edge (W);
\path[->] (X2) edge (W);
\path[<->] (X3) edge (W);
\path[<->] (X4) edge (W);
\path[<-] (X5) edge (W1);
\path[<-] (X6) edge (W2);
\end{tikzpicture}
\end{tabular}
\end{tabular}\caption{Conditioning the path diagram to the left on the node $A$ results in the diagram to the right.}\label{fig:illustration}
\end{figure}
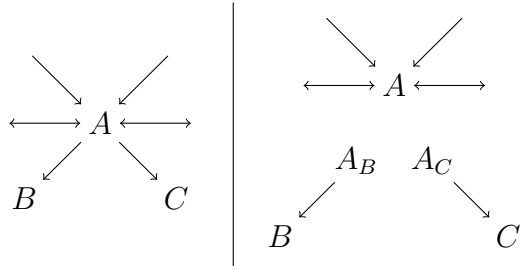

\begin{theorem}\label{the:condpathroot}
Consider a path diagram conditioned on a set of nodes $S$. Let $Z = S \cup S'$. Let $\Pi_{X:Y}$ denote all the $Z$-open paths from $X$ to $Y$. Suppose that no path in $\Pi_{X:Y}$ has colliders. Suppose that all the paths in $\Pi_{X:Y}$ have a subpath $X_m \la \cdots \la X_2 \la X_1 \ra X_{m+1} \ra \cdots \ra X_{m+n}$ or $X_1 = X \ra X_{2} \ra \cdots \ra X_{m+n}$. Suppose that there is no $Z$-open route $X_i \ra A \oo \cdots \oo B \oa X_i$ with $i>1$. Moreover, let $Z_i^i=Z_i \cup Z^i$ and $Z_{1:a}^{1:b} = Z_1 \cup \cdots \cup Z_a \cup Z^1 \cup \cdots \cup Z^b$ where
\begin{itemize}
\item $Z^i = \{W_1, W_2, \ldots\}$ is a subset of $Z \setminus Z^{1:i-1}_{1:i-1}$ such that each $W_j$ is $(Z^{1:i-1}_{1:i-1} \cup W_{1:j-1})$-connected to $X_i$ through $Pa(X_i) \cup Sp(X_i)$ by a path that does not contain any node that is in some path in $\Pi_{X:Y}$ except $X_i$, and
\item $Z_i = \{W_1, W_2, \ldots\}$ is a subset of $Z \setminus Z^{1:i}_{1:i-1}$ such that each $W_j$ is $(Z^{1:i}_{1:i-1} \cup W_{1:j-1})$-connected to $X_i$ through $Ch(X_i)$ by a path that does not contain any node that is in some path in $\Pi_{X:Y}$ except $X_i$.
\end{itemize}
Finally, let $Z \setminus Z^{1:m+n}_{1:m+n} = \{W_1, W_2, \ldots\}$ and suppose that $X \ci W_j | Z^{1:m+n}_{1:m+n} \cup W_{1:j-1}$ or $Y \ci W_j | Z^{1:m+n}_{1:m+n} \cup W_{1:j-1}$. Then,
\[
\sigma_{X Y \cdot Z} = \sigma_{X Y} \frac{\sigma^2_{X_1 \cdot Z_1^1}}{\sigma^2_{X_1}} \prod_{i=2}^{m+n} \frac{\sigma^2_{X_i \cdot Z_{1:i}^{1:i}}}{\sigma^2_{X_i \cdot Z_{1:i-1}^{1:i}}}.
\]
\end{theorem}

\begin{theorem}\label{the:condpathnonroot}
Consider the same assumptions as in Theorem \ref{the:condpathroot} with the only difference that all the paths in $\Pi_{X:Y}$ have now a subpath $X_m \la \cdots \la X_2 \la X_1 \aa X_{m+1} \ra \cdots \ra X_{m+n}$ or $X_1 \aa X_{2} \ra \cdots \ra X_{m+n}$ or $\oa X_1 \ra \cdots \ra X_{m+n}$.\footnote{In the third subpath type, the predecessor of $X_1$ does not have to be the same in every path in $\Pi_{X:Y}$. It just has to reach $X_1$ through an edge $\ra$ or $\aa$ in every path in $\Pi_{X:Y}$.\label{foo:subpath}} Then,
\[
\sigma_{X Y \cdot Z} = \sigma_{X Y} \prod_{i=1}^{m+n} \frac{\sigma^2_{X_i \cdot Z_{1:i}^{1:i}}}{\sigma^2_{X_i \cdot Z_{1:i-1}^{1:i}}}.
\]
\end{theorem}

Like Theorems \ref{the:condpath1} and \ref{the:condpath2}, the two theorems above show that the partial covariance of two random variables can be computed by correcting their covariance with the product of some partial variance ratios. This implies that the partial covariance inherits the salient feature of factorizing over the nodes and edges in the paths between the two variables of interest. The two theorems above also imply that conditioning does not change the sign of the covariance, as stated in the following immediate corollary.

\begin{corollary}
Suppose that two sets of nodes $S_1$ and $S_2$ satisfy the assumptions in Theorem \ref{the:condpathroot} or \ref{the:condpathnonroot}. Then, $sign(\sigma_{XY}) = sign(\sigma_{XY \cdot S_1})=sign(\sigma_{XY \cdot S_2})$.
\end{corollary}

For the path diagrams that satisfy the conditions in Theorem \ref{the:condpathroot} or \ref{the:condpathnonroot}, the corollary above implies that conditioning does not change the sign of the causal effect of $X$ on $Y$, and that Simpson's paradox cannot occur.

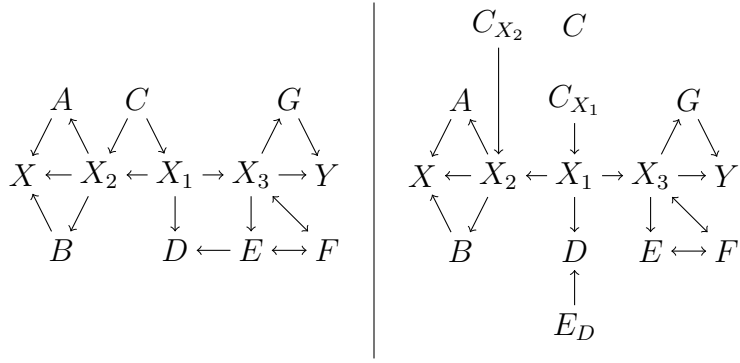
\begin{figure}
\begin{tabular}{c|c}
\begin{tabular}{c}
\begin{tikzpicture}[inner sep=1mm]
\node at (0,0) (X) {$X$};
\node at (1,0) (X1) {$X_2$};
\node at (2,0) (X2) {$X_1$};
\node at (3,0) (X3) {$X_3$};
\node at (4,0) (Y) {$Y$};
\node at (3.5,1) (G) {$G$};
\node at (0.5,1) (A) {$A$};
\node at (0.5,-1) (B) {$B$};
\node at (4,-1) (C) {$F$};
\node at (1.5,1) (S1) {$C$};
\node at (2,-1) (S2) {$D$};
\node at (3,-1) (S3) {$E$};
\path[<-] (X) edge (X1);
\path[<-] (X1) edge (X2);
\path[->] (X2) edge (X3);
\path[->] (X3) edge (Y);
\path[<-] (X) edge (A);
\path[<-] (A) edge (X1);
\path[<-] (B) edge (X1);
\path[->] (B) edge (X);
\path[<->] (X3) edge (C);
\path[<->] (C) edge (S3);
\path[->] (S1) edge (X1);
\path[->] (S1) edge (X2);
\path[->] (X2) edge (S2);
\path[->] (S3) edge (S2);
\path[<-] (S3) edge (X3);
\path[->] (X3) edge (G);
\path[->] (G) edge (Y);
\end{tikzpicture}
\end{tabular}
&
\begin{tabular}{c}
\begin{tikzpicture}[inner sep=1mm]
\node at (0,0) (X) {$X$};
\node at (1,0) (X1) {$X_2$};
\node at (2,0) (X2) {$X_1$};
\node at (3,0) (X3) {$X_3$};
\node at (4,0) (Y) {$Y$};
\node at (3.5,1) (G) {$G$};
\node at (0.5,1) (A) {$A$};
\node at (0.5,-1) (B) {$B$};
\node at (4,-1) (C) {$F$};
\node at (2,2) (D) {$C$};
\node at (1,2) (S1X1) {$C_{X_2}$};
\node at (2,1) (S1X2) {$C_{X_1}$};
\node at (2,-1) (S2) {$D$};
\node at (3,-1) (S3) {$E$};
\node at (2,-2) (S3S2) {$E_{D}$};
\path[<-] (X) edge (X1);
\path[<-] (X1) edge (X2);
\path[->] (X2) edge (X3);
\path[->] (X3) edge (Y);
\path[<-] (X) edge (A);
\path[<-] (A) edge (X1);
\path[<-] (B) edge (X1);
\path[->] (B) edge (X);
\path[<->] (X3) edge (C);
\path[<->] (C) edge (S3);
\path[->] (S1X1) edge (X1);
\path[->] (S1X2) edge (X2);
\path[->] (X2) edge (S2);
\path[->] (S3S2) edge (S2);
\path[<-] (S3) edge (X3);
\path[->] (X3) edge (G);
\path[->] (G) edge (Y);
\end{tikzpicture}
\end{tabular}
\end{tabular}\caption{Left: Path diagram where Theorem \ref{the:condpathroot} can be applied to compute $\sigma_{XY \cdot C D E}$. Right: The path diagram to the left conditioned on $\{C, D, E\}$.}\label{fig:example}
\end{figure}

We illustrate Theorem \ref{the:condpathroot} with the following example.

\begin{example}
Consider the path diagram to the left in Figure \ref{fig:example}. Say that we want to compute $\sigma_{XY \cdot S}$ with $S=\{C,D,E\}$. The path diagram conditioned on $S$ can be seen to the right in Figure \ref{fig:example}. Then, $S'=\{C_{X_1}, C_{X_2}, E_D\}$ and $Z = S \cup S' = \{C,D,E,C_{X_1}, C_{X_2}, E_D\}$. As discussed before, $\sigma_{XY \cdot S}$ in the original diagram coincides with $\sigma_{X Y \cdot Z}$ in the conditional diagram. Now, note that $Z^1=\{C_{X_1}\}$, $Z_1=\{D,E_D\}$, $Z^2=\{C_{X_2}\}$, $Z_2=\emptyset$, $Z^3=\emptyset$, and $Z_3=\{E\}$. Then, Theorem \ref{the:condpathroot} gives
\[
\sigma_{X Y \cdot Z} = \sigma_{X Y} \frac{\sigma^2_{X_1 \cdot C_{X_1} D E_D}}{\sigma^2_{X_1}} \frac{\sigma^2_{X_2 \cdot C_{X_1} D E_D C_{X_2}}}{\sigma^2_{X_2 \cdot C_{X_1} D E_D C_{X_2}}} \frac{\sigma^2_{X_3 \cdot C_{X_1} D E_D C_{X_2} E}}{\sigma^2_{X_3 \cdot C_{X_1} D E_D C_{X_2}}}.
\]
\end{example}

\begin{figure}
\begin{tabular}{c|c}
\begin{tabular}{c}
\begin{tikzpicture}[inner sep=1mm]
\node at (0,0) (X) {$X$};
\node at (1,0) (X1) {$X_1$};
\node at (2,0) (X2) {$X_2$};
\node at (3,0) (X3) {$X_3$};
\node at (4,0) (Y) {$Y$};
\node at (0.5,1) (A) {$A$};
\node at (0.5,-1) (B) {$B$};
\node at (1.5,1) (C) {$C$};
\node at (2.5,2) (D) {$D$};
\node at (2.5,1) (E) {$E$};
\node at (3.5,-1) (F) {$F$};
\path[->] (X) edge (X1);
\path[->] (X1) edge (X2);
\path[->] (X2) edge (X3);
\path[->] (X3) edge (Y);
\path[<-] (X) edge (A);
\path[<->] (A) edge (X1);
\path[->] (B) edge (X1);
\path[<-] (B) edge (X);
\path[->] (C) edge (X1);
\path[->] (C) edge (X2);
\path[->] (E) edge (X2);
\path[->] (E) edge (X3);
\path[->] (D) edge (E);
\path[->] (D) edge (Y);
\path[->] (X3) edge (F);
\path[->] (F) edge (Y);
\end{tikzpicture}
\end{tabular}
&
\begin{tabular}{c}
\begin{tikzpicture}[inner sep=1mm]
\node at (0,0) (X) {$X$};
\node at (1,0) (X1) {$X_1$};
\node at (2,0) (X2) {$X_2$};
\node at (3,0) (X3) {$X_3$};
\node at (4,0) (Y) {$Y$};
\node at (0.5,1) (A) {$A$};
\node at (0.5,-1) (B) {$B$};
\node at (2,2) (C) {$C$};
\node at (1,2) (S1X1) {$C_{X_1}$};
\node at (2,1) (S1X2) {$C_{X_2}$};
\node at (4,2) (D) {$D$};
\node at (4,1) (DY) {$D_Y$};
\node at (3,2) (DE) {$D_E$};
\node at (3,1) (E) {$E$};
\node at (3.5,-1) (F) {$F$};
\path[->] (X) edge (X1);
\path[->] (X1) edge (X2);
\path[->] (X2) edge (X3);
\path[->] (X3) edge (Y);
\path[->] (X) edge (A);
\path[<->] (A) edge (X1);
\path[->] (B) edge (X1);
\path[<-] (B) edge (X);
\path[->] (S1X1) edge (X1);
\path[->] (S1X2) edge (X2);
\path[->] (E) edge (X2);
\path[->] (E) edge (X3);
\path[->] (DY) edge (Y);
\path[->] (DE) edge (E);
\path[->] (X3) edge (F);
\path[->] (F) edge (Y);
\end{tikzpicture}
\end{tabular}
\end{tabular}\caption{Left: Path diagram where Theorem \ref{the:condpathnonroot} can be applied to compute $\sigma_{XY \cdot C D}$. Right: The path diagram to the left conditioned on $\{C, D\}$.}\label{fig:example2}
\end{figure}
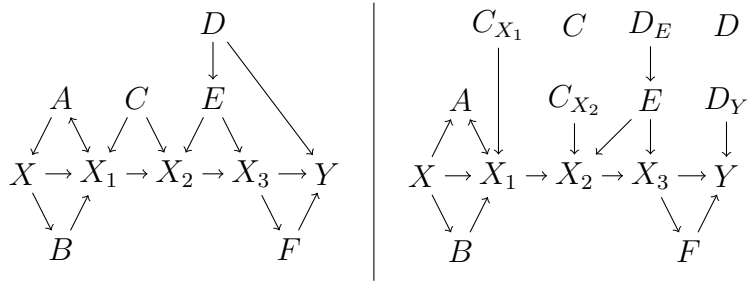

We illustrate Theorem \ref{the:condpathnonroot} with the following example.

\begin{example}
Consider the path diagram to the left in Figure \ref{fig:example2}. Say that we want to compute $\sigma_{XY \cdot S}$ with $S=\{C,D\}$. The path diagram conditioned on $S$ can be seen to the right in Figure \ref{fig:example2}. Then, $S'=\{C_{X_1}, C_{X_2}, D_E\}$ and $Z = S \cup S' = \{C,D,C_{X_1}, C_{X_2}, D_E, D_Y\}$. As discussed before, $\sigma_{XY \cdot S}$ in the original diagram coincides with $\sigma_{X Y \cdot Z}$ in the conditional diagram. Now, note that $Z^1=\{C_{X_1}\}$, $Z^2=\{C_{X_2}, D_E\}$, and $Z_1=Z_2=Z_3=Z^3=\emptyset$. Then, Theorem \ref{the:condpathnonroot} gives
\[
\sigma_{X Y \cdot Z} = \sigma_{X Y} \frac{\sigma^2_{X_1 \cdot C_{X_1}}}{\sigma^2_{X_1 \cdot C_{X_1}}} \frac{\sigma^2_{X_2 \cdot C_{X_1} C_{X_2} D_E}}{\sigma^2_{X_2 \cdot C_{X_1} C_{X_2} D_E}} \frac{\sigma^2_{X_3 \cdot C_{X_1} C_{X_2} D_E}}{\sigma^2_{X_3 \cdot C_{X_1} C_{X_2} D_E}}.
\]
\end{example}

\begin{figure}
\begin{tabular}{c|c}
\begin{tabular}{c}
\begin{tikzpicture}[inner sep=1mm]
\node at (0,0) (X) {$X_1$};
\node at (1,0) (X1) {$X_2$};
\node at (2,0) (X2) {$X_3$};
\node at (3,0) (Y) {$X_4$};
\node at (0.5,1) (A) {$A$};
\node at (1.5,-1) (B) {$C$};
\node at (2.5,1) (C) {$B$};
\node at (3,-1) (D) {$D$};
\path[->] (X) edge (X1);
\path[->] (X1) edge (X2);
\path[->] (X2) edge (Y);
\path[<-] (X1) edge (A);
\path[<-] (A) edge (X);
\path[->] (C) edge (X2);
\path[->] (C) edge (Y);
\path[->] (X1) edge (B);
\path[->] (B) edge (X2);
\path[->] (Y) edge (D);
\end{tikzpicture}
\end{tabular}
&
\begin{tabular}{c}
\begin{tikzpicture}[inner sep=1mm]
\node at (0,0) (X) {$X_1$};
\node at (1,0) (X1) {$X_2$};
\node at (2,0) (X2) {$X_3$};
\node at (3,0) (Y) {$X_4$};
\node at (0,1) (A) {$A$};
\node at (1,1) (AX2) {$A_{X_2}$};
\node at (1.5,-1) (B) {$C$};
\node at (2.5,2) (C) {$B$};
\node at (2,1) (CX2) {$B_{X_3}$};
\node at (3,1) (CY) {$B_{X_4}$};
\node at (3,-1) (D) {$D$};
\path[->] (X) edge (X1);
\path[->] (X1) edge (X2);
\path[->] (X2) edge (Y);
\path[<-] (X1) edge (AX2);
\path[<-] (A) edge (X);
\path[->] (X1) edge (B);
\path[->] (B) edge (X2);
\path[->] (CX2) edge (X2);
\path[->] (CY) edge (Y);
\path[->] (Y) edge (D);
\end{tikzpicture}
\end{tabular}
\end{tabular}\caption{Left: Path diagram where Theorems \ref{the:condpathroot} and \ref{the:condpathnonroot} can be combined to compute $\sigma_{X_1 X_4 \cdot A B D}$. Left: The path diagram to the left conditioned on $\{A, B, D\}$.}\label{fig:example3}
\end{figure}
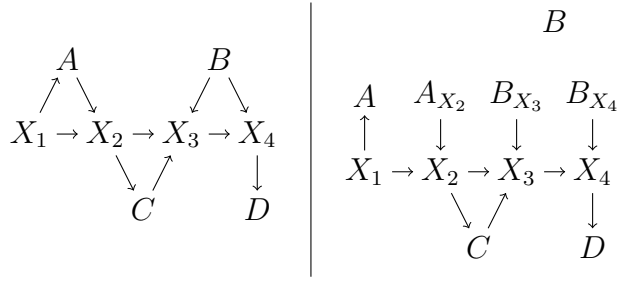

Theorems \ref{the:condpathroot} and \ref{the:condpathnonroot} can be extended to when all the paths from $X$ to $Y$ in the conditional path diagram share more than one subpath. For instance, consider the path diagram to the left in Figure \ref{fig:example3}. Say that we want to compute $\sigma_{X_1 X_4 \cdot S}$ with $S=\{A, B, D\}$. The path diagram conditioned on $S$ can be seen to the right in Figure \ref{fig:example3}. Then, $S'=\{A_{X_2}, B_{X_3}, B_{X_4}\}$ and $Z = S \cup S' = \{A, B, D, A_{X_2}, B_{X_3}, B_{X_4}\}$. As discussed before, $\sigma_{X_1 X_4 \cdot S}$ in the original diagram coincides with $\sigma_{X_1 X_4 \cdot Z}$ in the conditional diagram. Now, note that $Z_1=\{A\}$, $Z^2=\{A_{X_2}\}$, $Z^3=\{B_{X_3}\}$, $Z^4=\{B_{X_4}\}$, $Z_4=\{D\}$, and $Z^1=Z_2=Z_3=\emptyset$. Note also that the conditional diagram has two $Z$-open paths from $X_1$ to $X_4$, which share two subpaths: $X_1 \ra X_2$ and $\ra X_3 \ra X_4$. Therefore, neither Theorem \ref{the:condpathroot} nor \ref{the:condpathnonroot} applies. However, applying Theorem \ref{the:condpathroot} followed by Theorem \ref{the:condpathnonroot} gives
\begin{equation}\label{eq:example3}
\sigma_{X_1 X_4 \cdot Z} = \sigma_{X_1 X_4} \frac{\sigma^2_{X_1 \cdot A}}{\sigma^2_{X_1}} \frac{\sigma^2_{X_2 \cdot A A_{X_2}}}{\sigma^2_{X_2 \cdot A A_{X_2}}} \frac{\sigma^2_{X_3 \cdot A A_{X_2} B_{X_3}}}{\sigma^2_{X_3 \cdot A A_{X_2} B_{X_3}}} \frac{\sigma^2_{X_4 \cdot A A_{X_2} B_{X_3} B_{X_4} D}}{\sigma^2_{X_4 \cdot A A_{X_2} B_{X_3} B_{X_4}}}.
\end{equation}
The proof of correctness of the previous expression is simply a concatenation of the proofs of Theorems \ref{the:condpathroot} and \ref{the:condpathnonroot}. We omit the details. An alternative way of answering the previous query is by first absorbing the subpath $X_2 \ra C \ra X_3$ into the subpath $X_2 \ra X_3$. Now, there is only one shared subpath in the path diagram conditioned on $S$, namely $X_1 \ra X_2 \ra X_3 \ra X_4$. Then, Theorem \ref{the:condpathroot} gives Equation \ref{eq:example3}. This absorption trick is always possible when there are several shared subpaths. We omit the details.

%% file: AppendixF.tex
\section*{Appendix E: Proofs of Appendix D}

We start with some lemmas stating some auxiliary results. Recall from Footnote \ref{foo:subpath} that when we say that every path has a subpath $A \oa$, we do not mean that the successor of $A$ is the same in every path. We mean that the successor is reached through an edge $\ra$ or $\aa$ in every path. Given a route $\rho_{X:Y}$ from $X$ to $Y$, we let $\rho_{X:A}$ denote the subroute of $\rho_{X:Y}$ from $X$ to $A$. Given two routes $\rho_{X:A}$ and $\rho_{A:Y}$, we let $\rho_{X:A} \cup \rho_{A:Y}$ denote the route from $X$ to $Y$ resulting from concatenating $\rho_{X:A}$ and $\rho_{A:Y}$. Finally, the path corresponding to a $Z$-open route from $X$ to $Y$ is a $Z$-open path from $X$ to $Y$ whose edges are a subset of the edges in the route. Such a path always exists by Lemma \ref{lem:pathroute}.

\begin{lemma}\label{lem:pathroute}
There is a $Z$-open route from $X$ to $Y$ if and only if there is a $Z$-open path from $X$ to $Y$. Moreover, the path and the route can be chosen such that the edges in the former are a subset of the edges in the latter.
\end{lemma}

\begin{proof}
Let $\pi_{X:Y}$ be a $Z$-open path from $X$ to $Y$. For every subpath $A \oa C \ao B$ of $\pi_{X:Y}$ such that $C \notin Z$, do the following. First, find a path $C \ra X_1 \ra \cdots \ra X_n$ with $X_n \in Z$ and $X_1, \ldots, X_{n-1} \notin Z$. Such a path must exist for $\pi_{X:Y}$ to be $Z$-open. Second, replace $A \oa C \ao B$ with $A \oa C \ra X_1 \ra \cdots \ra X_n \la \cdots \la X_1 \la C \ao B$. The result is the desired route.

Let $\rho_{X:Y}$ be a $Z$-open route from $X$ to $Y$. Repeat the following two steps while possible. The result is the desired path. First, choose a node $A$ that occurs several times in $\rho_{X:Y}$. Let $A_1$ and $A_2$ denote the first and the last occurrences of $A$ in $\rho_{X:Y}$. Assume without loss of generality that no node in $\rho_{X:A_1}$ occurs in $\rho_{A_2:Y}$ except $A$. Second, consider the following cases.
\begin{itemize}
\item If $\rho_{X:Y}$ is $X \oo \cdots \la A_1 \oo \cdots \oo A_2 \oo \cdots \oo Y$, then replace it with $\rho_{X:A_1} \cup \rho_{A_2:Y}$.

\item If $\rho_{X:Y}$ is $X \oo \cdots \oa A_1 \oo \cdots \oo A_2 \ra \cdots \oo Y$, then replace it with $\rho_{X:A_1} \cup \rho_{A_2:Y}$.

\item If $\rho_{X:Y}$ is $X \oo \cdots \oa A_1 \oo \cdots \oo A_2 \ao \cdots \oo Y$, then replace it with $\rho_{X:A_1} \cup \rho_{A_2:Y}$. Note that $A$ or some descendant of $A$ must be in $Z$ for the original $\rho_{X:Y}$ to be $Z$-open.
\end{itemize}
\end{proof}

\begin{lemma}\label{lem:aux1}
Consider a path diagram. Let $X$, $Y$, $R$ and $W$ be nodes and $Z$ a set of nodes. If $X \ci W | Z \cup R$ and $Y \ci W | Z \cup R$ and $X \ci Y | Z \cup R$, then
\[
\sigma_{XY \cdot ZW} =  \sigma_{XY \cdot Z} \frac{\sigma^2_{R \cdot ZW}}{\sigma^2_{R \cdot Z}}.
\]
\end{lemma}

\begin{proof}
Note that $X \ci W | Z \cup R$ implies that
\[
0 = \sigma_{XW \cdot ZR} = \sigma_{XW \cdot Z} - \frac{\sigma_{XR \cdot Z} \sigma_{RW \cdot Z}}{\sigma^2_{R \cdot Z}}
\]
which implies that $\sigma_{XW \cdot Z} = \delta_{XR \cdot Z} \sigma_{RW \cdot Z}$ where $\delta_{XR \cdot Z} = \sigma_{XR \cdot Z} / \sigma^2_{R \cdot Z}$. Likewise, $Y \ci W | Z \cup R$ implies that $\sigma_{YW \cdot Z} = \delta_{YR \cdot Z} \sigma_{RW \cdot Z}$ where $\delta_{YR \cdot Z} = \sigma_{YR \cdot Z} / \sigma^2_{R \cdot Z}$. Likewise, $X \ci Y | Z \cup R$ implies that $\sigma_{XY \cdot Z} = \delta_{XR \cdot Z} \delta_{YR \cdot Z} \sigma^2_{R \cdot Z}$. Therefore,
\begin{align*}
\sigma_{XY \cdot ZW} &= \sigma_{XY \cdot Z} - \frac{\sigma_{XW \cdot Z} \sigma_{WY \cdot Z}}{\sigma^2_{W \cdot Z}}\\
& = \delta_{XR \cdot Z} \delta_{YR \cdot Z} \sigma^2_{R \cdot Z} - \frac{\delta_{XR \cdot Z} \sigma_{RW \cdot Z} \delta_{YR \cdot Z} \sigma_{RW \cdot Z}}{\sigma^2_{W \cdot Z}}\\
& = \delta_{XR \cdot Z} \delta_{YR \cdot Z} \Big( \sigma^2_{R \cdot Z} - \frac{\sigma_{RW \cdot Z} \sigma_{RW \cdot Z}}{\sigma^2_{W \cdot Z}} \Big)\\
& = \delta_{XR \cdot Z} \delta_{YR \cdot Z} \sigma^2_{R \cdot ZW} =  \sigma_{XY \cdot Z} \frac{\sigma^2_{R \cdot ZW}}{\sigma^2_{R \cdot Z}}.
\end{align*}
\end{proof}

\begin{lemma}\label{lem:aux2}
Consider a path diagram. Let $X$, $Y$ and $W$ be nodes and $Z$ a set of nodes. If $Y \ci W | Z \cup X$, then
\[
\sigma_{XY \cdot ZW} =  \sigma_{XY \cdot Z} \frac{\sigma^2_{X \cdot ZW}}{\sigma^2_{X \cdot Z}}.
\]
\end{lemma}

\begin{proof}
Note that $Y \ci W | Z \cup X$ implies that
\[
0 = \sigma_{YW \cdot ZX} = \sigma_{YW \cdot Z} - \frac{\sigma_{YX \cdot Z} \sigma_{XW \cdot Z}}{\sigma^2_{X \cdot Z}}
\]
which implies that
\[
\sigma_{YW \cdot Z} = \frac{\sigma_{YX \cdot Z} \sigma_{XW \cdot Z}}{\sigma^2_{X \cdot Z}}.
\]
Therefore,
\begin{align*}
\sigma_{XY \cdot ZW} &= \sigma_{XY \cdot Z} - \frac{\sigma_{XW \cdot Z} \sigma_{WY \cdot Z}}{\sigma^2_{W \cdot Z}} = \sigma_{XY \cdot Z} - \frac{\sigma_{XW \cdot Z} \sigma_{YX \cdot Z} \sigma_{XW \cdot Z}}{\sigma^2_{W \cdot Z} \sigma^2_{X \cdot Z}}\\
&= \sigma_{XY \cdot Z} \Big( 1 - \frac{\sigma_{XW \cdot Z} \sigma_{XW \cdot Z}}{\sigma^2_{W \cdot Z} \sigma^2_{X \cdot Z}} \Big) = \frac{\sigma_{XY \cdot Z}}{\sigma^2_{X \cdot Z}} \Big( \sigma^2_{X \cdot Z} - \frac{\sigma_{XW \cdot Z} \sigma_{XW \cdot Z}}{\sigma^2_{W \cdot Z}} \Big)\\
&= \sigma_{XY \cdot Z} \frac{\sigma^2_{X \cdot ZW}}{\sigma^2_{X \cdot Z}}.
\end{align*}
\end{proof}

\begin{lemma}\label{lem:aux3}
Consider a path diagram. Let $X$, $Y$ and $W$ be nodes and $Z$ a set of nodes. If $X \ci W | Z$ or $Y \ci W | Z$, then $\sigma_{XY \cdot ZW} =  \sigma_{XY \cdot Z}$.
\end{lemma}

\begin{proof}
Note that $X \ci W | Z$ and $Y \ci W | Z$ imply $\sigma_{XW \cdot Z}=0$ and $\sigma_{YW \cdot Z}=0$, respectively. Thus, if $X \ci W | Z$ or $Y \ci W | Z$ then
\[
\sigma_{XY \cdot ZW} = \sigma_{XY \cdot Z} - \frac{\sigma_{XW \cdot Z} \sigma_{WY \cdot Z}}{\sigma^2_{W \cdot Z}} = \sigma_{XY \cdot Z}.
\]
\end{proof}

\begin{lemma}\label{lem:root1}
Consider a path diagram. Let $\Pi_{X:Y}$ denote all the $Z$-open paths from $X$ to $Y$. Suppose that no path in $\Pi_{X:Y}$ has colliders. Suppose that all the paths in $\Pi_{X:Y}$ have a subpath $\la R \ra$ or $R = X \ra$. Let $W$ be a node that is $Z$-connected to $R$ by a path that does not contain any node that is in some path in $\Pi_{X:Y}$ except $R$. If $Pa(W) \cup Sp(W)= \emptyset$, then
\[
\sigma_{XY \cdot ZW} =  \sigma_{XY \cdot Z} \frac{\sigma^2_{R \cdot ZW}}{\sigma^2_{R \cdot Z}}.
\]
\end{lemma}

\begin{proof}
Consider first the case where $\la R \ra$ is a subpath of every path in $\Pi_{X:Y}$. Assume to the contrary that $X \nci W | Z \cup R$ and let $\rho_{X:W}$ be a $(Z \cup R)$-open {\bf route}. Follow $\rho_{X:W}$ until reaching $R$ or $W$, and let $\pi_{X:Y} \in \Pi_{X:Y}$.
\begin{itemize}
\item If $R$ is reached first, then consider the first occurrence of $R$ in $\rho_{X:W}$ and note that $\rho_{X:R} \cup \pi_{R:Y}$ is a $Z$-open route. Moreover, it does not contain any edge $\la R$ that is in some path in $\Pi_{X:Y}$ because, otherwise, $\rho_{X:W}$ contains the edge ($\pi_{R:Y}$ cannot by definition) and, thus, it is not $(Z \cup R)$-open. Then, the path corresponding to $\rho_{X:R} \cup \pi_{R:Y}$ contradicts the assumptions in the lemma.

\item If $W$ is reached first, then let $\varrho_{R:W}$ denote a $Z$-open path that does not contain any node that is in some path in $\Pi_{X:Y}$ except $R$. Such a path exists by the assumptions in the lemma. Thus, $\rho_{X:W} \cup \varrho_{W:R} \cup \pi_{R:Y}$ is a $Z$-open route, because neither $R$ nor $W$ is a collider in it. The latter follows from the assumption that $Pa(W) \cup Sp(W)=\emptyset$. Moreover, the route does not contain any edge $\la R$ that is in some path in $\Pi_{X:Y}$ because, otherwise, $\rho_{X:W}$ contains the edge ($\varrho_{W:R}$ and $\pi_{R:Y}$ cannot by definition) and thus it reaches $R$ first. Then, the path corresponding to $\rho_{X:W} \cup \varrho_{W:R} \cup \pi_{R:Y}$ contradicts the assumptions in the lemma.
\end{itemize}
Consequently, $X \ci W | Z \cup R$. We can analogously prove that $Y \ci W | Z \cup R$. Now, assume to the contrary that $X \nci Y | Z \cup R$ and let $\rho_{X:Y}$ be a $(Z \cup R)$-open {\bf route}. Note that $R$ must be in $\rho_{X:Y}$ because, otherwise, $\rho_{X:Y}$ is $Z$-open and, thus, its corresponding path contradicts the assumptions in the lemma. Consider the first occurrence of $R$ in $\rho_{X:Y}$, and let $\pi_{X:Y} \in \Pi_{X:Y}$. Then, $\rho_{X:R} \cup \pi_{R:Y}$ is a $Z$-open route. Moreover, the route does not contain any edge $\la R$ that is in some path in $\Pi_{X:Y}$ because, otherwise, $\rho_{X:Y}$ contains the edge ($\pi_{R:Y}$ cannot by definition) and thus it is not $(Z \cup R)$-open. Then, the path corresponding to $\rho_{X:R} \cup \pi_{R:Y}$ contradicts the assumptions in the lemma. Consequently, $X \ci Y | Z \cup R$. Therefore, the desired result follows from Lemma \ref{lem:aux1}.

Finally, consider the case where $R = X \ra$ is a subpath of every path in $\Pi_{X:Y}$. Assume to the contrary that $Y \nci W | Z \cup X$ and let $\rho_{Y:W}$ be a $(Z \cup X)$-open {\bf route}. Follow $\rho_{Y:W}$ until reaching $X$ or $W$.
\begin{itemize}
\item If $X$ is reached first, then note that $\rho_{Y:X}$ does not contain any edge $\la X$ that is in some path in $\Pi_{Y:X}$ because, otherwise, $\rho_{Y:W}$ is not $(Z \cup X)$-open. Then, the path corresponding to $\rho_{X:Y}$ contradicts the assumptions in the lemma.

\item If $W$ is reached first, let $\varrho_{X:W}$ denote a $Z$-open path that does not contain any node that is in some path in $\Pi_{X:Y}$ except $X$. Such a path exists by the assumptions in the lemma. Then, $\varrho_{X:W} \cup \rho_{W:Y}$ is a $Z$-open route, because $W$ is not a collider in it due to the assumption that $Pa(W) \cup Sp(W)=\emptyset$. Moreover, the route does not contain any edge $X \ra$ that is in some path in $\Pi_{X:Y}$ because, otherwise, $\rho_{W:Y}$ contains the edge ($\varrho_{X:W}$ cannot by definition) and thus $\rho_{Y:W}$ reaches $X$ first. Then, the path corresponding to $\varrho_{X:W} \cup \rho_{W:Y}$ contradicts the assumptions in the lemma.
\end{itemize}
Consequently, $Y \ci W | Z \cup X$ and, thus, the desired result follows from Lemma \ref{lem:aux2}.
\end{proof}

\begin{lemma}\label{lem:root2}
Consider a path diagram. Let $\Pi_{X:Y}$ denote all the $Z$-open paths from $X$ to $Y$. Suppose that no path in $\Pi_{X:Y}$ has colliders. Suppose that all the paths in $\Pi_{X:Y}$ have a subpath $\la R \ra$ or $R = X \ra$. Let $W$ be a node that is $Z$-connected to $R$ by a path that does not contain any node that is in some path in $\Pi_{X:Y}$ except $R$. If $Ch(W)= \emptyset$ and $\Pi_{X:Y}$ are all the $(Z \cup W)$-open paths from $X$ to $Y$, then
\[
\sigma_{XY \cdot ZW} =  \sigma_{XY \cdot Z} \frac{\sigma^2_{R \cdot ZW}}{\sigma^2_{R \cdot Z}}.
\]
\end{lemma}

\begin{proof}
Consider first the case where $\la R \ra$ is a subpath of every path in $\Pi_{X:Y}$. Assume to the contrary that $X \nci W | Z \cup R$ and let $\rho_{X:W}$ be a $(Z \cup R)$-open {\bf route}. Follow $\rho_{X:W}$ until reaching $R$ or $W$, and let $\pi_{X:Y} \in \Pi_{X:Y}$.
\begin{itemize}
\item If $R$ is reached first, then consider the first occurrence of $R$ in $\rho_{X:W}$ and note that $\rho_{X:R} \cup \pi_{R:Y}$ is a $Z$-open route. Moreover, it does not contain any edge $\la R$ that is in some path in $\Pi_{X:Y}$ because, otherwise, $\rho_{X:W}$ contains the edge ($\pi_{R:Y}$ cannot by definition) and, thus, it is not $(Z \cup R)$-open. Then, the path corresponding to $\rho_{X:R} \cup \pi_{R:Y}$ contradicts the assumptions in the lemma.

\item If $W$ is reached first, then let $\varrho_{R:W}$ denote a $Z$-open path that does not contain any node that is in some path in $\Pi_{X:Y}$ except $R$. Such a path exists by the assumptions in the lemma. Thus, $\rho_{X:W} \cup \varrho_{W:R} \cup \pi_{R:Y}$ is a $(Z \cup W)$-open route from $X$ to $Y$, because $W$ is a collider in it whereas $R$ is not. The former follows from the assumption that $Ch(W)=\emptyset$. Moreover, the route does not contain any edge $\la R$ that is in some path in $\Pi_{X:Y}$ because, otherwise, $\rho_{X:W}$ contains the edge ($\varrho_{W:R}$ and $\pi_{R:Y}$ cannot by definition) and thus it reaches $R$ first. Then, the path corresponding to $\rho_{X:W} \cup \varrho_{W:R} \cup \pi_{R:Y}$ contradicts the assumptions in the lemma.
\end{itemize}
Consequently, $X \ci W | Z \cup R$. We can analogously prove that $Y \ci W | Z \cup R$. Now, assume to the contrary that $X \nci Y | Z \cup R$ and let $\rho_{X:Y}$ be a $(Z \cup R)$-open {\bf route}. Note that $R$ must be in $\rho_{X:Y}$ because, otherwise, $\rho_{X:Y}$ is $Z$-open and, thus, its corresponding path contradicts the assumptions in the lemma. Consider the first occurrence of $R$ in $\rho_{X:Y}$, and let $\pi_{X:Y} \in \Pi_{X:Y}$. Then, $\rho_{X:R} \cup \pi_{R:Y}$ is a $Z$-open route. Moreover, the route does not contain any edge $\la R$ that is in some path in $\Pi_{X:Y}$ because, otherwise, $\rho_{X:Y}$ contains the edge ($\pi_{R:Y}$ cannot by definition) and thus it is not $(Z \cup R)$-open. Then, the path corresponding to $\rho_{X:R} \cup \pi_{R:Y}$ contradicts the assumptions in the lemma. Consequently, $X \ci Y | Z \cup R$. Therefore, the desired result follows from Lemma \ref{lem:aux1}.

Finally, consider the case where $R = X \ra$ is a subpath of every path in $\Pi_{X:Y}$. Assume to the contrary that $Y \nci W | Z \cup X$ and let $\rho_{Y:W}$ be a $(Z \cup X)$-open {\bf route}. Follow $\rho_{Y:W}$ until reaching $X$ or $W$.
\begin{itemize}
\item If $X$ is reached first, then note that $\rho_{Y:X}$ does not contain any edge $\la X$ that is in some path in $\Pi_{Y:X}$ because, otherwise, $\rho_{Y:W}$ is not $(Z \cup X)$-open. Then, the path corresponding to $\rho_{X:Y}$ contradicts the assumptions in the lemma.

\item If $W$ is reached first, let $\varrho_{X:W}$ denote a $Z$-open path that does not contain any node that is in some path in $\Pi_{X:Y}$ except $X$. Such a path exists by the assumptions in the lemma. Then, $\varrho_{X:W} \cup \rho_{W:Y}$ is a $(Z \cup W)$-open route, because $W$ is a collider in it due to the assumption that $Ch(W)=\emptyset$. Moreover, the route does not contain any edge $X \ra$ that is in some path in $\Pi_{X:Y}$ because, otherwise, $\rho_{W:Y}$ contains the edge ($\varrho_{X:W}$ cannot by definition) and thus $\rho_{Y:W}$ reaches $X$ first. Then, the path corresponding to $\varrho_{X:W} \cup \rho_{W:Y}$ contradicts the assumptions in the lemma.
\end{itemize}
Consequently, $Y \ci W | Z \cup X$ and, thus, the desired result follows from Lemma \ref{lem:aux2}.
\end{proof}

\begin{lemma}\label{lem:nonroot1}
Consider a path diagram. Let $\Pi_{X:Y}$ denote all the $Z$-open paths from $X$ to $Y$. Suppose that no path in $\Pi_{X:Y}$ has colliders. Suppose that all the paths in $\Pi_{X:Y}$ have a subpath $\oa R \ra$ or $\oa Y = R$. Let $W$ be a node that is $Z$-connected to $R$ through $Pa(R) \cup Sp(R)$ by a path that does not contain any node that is in some path in $\Pi_{X:Y}$ except $R$. If $Pa(W) \cup Sp(W)= \emptyset$ and $W$ is not $Z$-connected to $R$ through $Ch(R)$, then
\[
\sigma_{XY \cdot ZW} =  \sigma_{XY \cdot Z}.
\]
\end{lemma}

\begin{proof}
Consider first the case where $\oa R \ra$ is a subpath of every path in $\Pi_{X:Y}$. Assume to the contrary that $X \nci W | Z$ and let $\rho_{X:W}$ be a $Z$-open {\bf route}. Follow $\rho_{X:W}$ until reaching $R$ or $W$, and let $\pi_{X:Y} \in \Pi_{X:Y}$.
\begin{itemize}
\item If $R$ is reached first, then consider the first occurrence of $R$ in $\rho_{X:W}$ and note that $\rho_{X:R}$ ends with an edge $\oa R$ because, otherwise, $\rho_{X:R} \cup \pi_{R:Y}$ is a $Z$-open route that has a subroute $\la R \ra$ and, thus, its corresponding path contradicts the assumptions in the lemma. Then, $\rho_{R:W}$ must start with an edge $R \ra$ for $\rho_{X:W}$ to be $Z$-open. However, this contradicts the assumption that $W$ is not $Z$-connected to $R$ through $Ch(R)$.

\item If $W$ is reached first, then let $\varrho_{R:W}$ denote a $Z$-open path that does not contain any node that is in some path in $\Pi_{X:Y}$ except $R$. Such a path exists by the assumptions in the lemma. Thus, $\rho_{X:W} \cup \varrho_{W:R} \cup \pi_{R:Y}$ is a $Z$-open route, because neither $R$ nor $W$ is a collider in it. The latter follows from the assumption that $Pa(W) \cup Sp(W)=\emptyset$. Moreover, the route does not contain any edge $\oa R$ that is in some path in $\Pi_{X:Y}$ because, otherwise, $\rho_{X:W}$ contains the edge ($\varrho_{W:R}$ and $\pi_{R:Y}$ cannot by definition) and thus it reaches $R$ first. Then, the path corresponding to $\rho_{X:W} \cup \varrho_{W:R} \cup \pi_{R:Y}$ contradicts the assumptions in the lemma.
\end{itemize}
Consequently, $X \ci W | Z$. When $\oa Y = R$ is a subpath of every path in $\Pi_{X:Y}$, we can prove that $X \ci W | Z$ much in the same way. Consequently, $X \ci W | Z$ in either case and, thus, the desired result follows from Lemma \ref{lem:aux3}.
\end{proof}

\begin{lemma}\label{lem:nonroot2}
Consider a path diagram. Let $\Pi_{X:Y}$ denote all the $Z$-open paths from $X$ to $Y$. Suppose that no path in $\Pi_{X:Y}$ has colliders. Suppose that all the paths in $\Pi_{X:Y}$ have a subpath $\oa R \ra$ or $\oa Y = R$. Let $W$ be a node that is $Z$-connected to $R$ through $Pa(R) \cup Sp(R)$ by a path that does not contain any node that is in some path in $\Pi_{X:Y}$ except $R$. If $Ch(W)= \emptyset$, and $W$ is not $Z$-connected to $R$ through $Ch(R)$, and $\Pi_{X:Y}$ are all the $(Z \cup W)$-open paths from $X$ to $Y$, then
\[
\sigma_{XY \cdot ZW} =  \sigma_{XY \cdot Z}.
\]
\end{lemma}

\begin{proof}
Consider first the case where $\oa R \ra$ is a subpath of every path in $\Pi_{X:Y}$. Assume to the contrary that $X \nci W | Z$ and let $\rho_{X:W}$ be a $Z$-open {\bf route}. Follow $\rho_{X:W}$ until reaching $R$ or $W$, and let $\pi_{X:Y} \in \Pi_{X:Y}$.
\begin{itemize}
\item If $R$ is reached first, then consider the first occurrence of $R$ in $\rho_{X:W}$ and note that $\rho_{X:R}$ ends with an edge $\oa R$ because, otherwise, $\rho_{X:R} \cup \pi_{R:Y}$ is a $Z$-open route that has a subroute $\la R \ra$ and, thus, its corresponding path contradicts the assumptions in the lemma. Then, $\rho_{R:W}$ must start with an edge $R \ra$ for $\rho_{X:W}$ to be $Z$-open. However, this contradicts the assumption that $W$ is not $Z$-connected to $R$ through $Ch(R)$.

\item If $W$ is reached first, then let $\varrho_{R:W}$ denote a $Z$-open path that does not contain any node that is in some path in $\Pi_{X:Y}$ except $R$. Such a path exists by the assumptions in the lemma. Thus, $\rho_{X:W} \cup \varrho_{W:R} \cup \pi_{R:Y}$ is a $(Z \cup W)$-open route, because $W$ is a collider in it whereas $R$ is not. The former follows from the assumption that $Ch(W)=\emptyset$. Moreover, the route does not contain any edge $\oa R$ that is in some path in $\Pi_{X:Y}$ because, otherwise, $\rho_{X:W}$ contains the edge ($\varrho_{W:R}$ and $\pi_{R:Y}$ cannot by definition) and thus it reaches $R$ first. Then, the path corresponding to $\rho_{X:W} \cup \varrho_{W:R} \cup \pi_{R:Y}$ contradicts the assumptions in the lemma.
\end{itemize}
Consequently, $X \ci W | Z$. When $\oa Y = R$ is a subpath of every path in $\Pi_{X:Y}$, we can prove that $X \ci W | Z$ much in the same way. Consequently, $X \ci W | Z$ in either case and, thus, the desired result follows from Lemma \ref{lem:aux3}.
\end{proof}

\begin{lemma}\label{lem:nonroot3}
Consider a path diagram. Let $\Pi_{X:Y}$ denote all the $Z$-open paths from $X$ to $Y$. Suppose that no path in $\Pi_{X:Y}$ has colliders. Suppose that all the paths in $\Pi_{X:Y}$ have a subpath $\oa R \ra$ or $\oa Y = R$. Let $W$ be a node that is $Z$-connected to $R$ through $Ch(R)$ by a path that does not contain any node that is in some path in $\Pi_{X:Y}$ except $R$. If $Pa(W) \cup Sp(W)= \emptyset$ and $W$ is not $Z$-connected to $R$ through $Pa(R) \cup Sp(R)$, then
\[
\sigma_{XY \cdot ZW} =  \sigma_{XY \cdot Z} \frac{\sigma^2_{R \cdot ZW}}{\sigma^2_{R \cdot Z}}.
\]
\end{lemma}

\begin{proof}
Consider first the case where $\oa R \ra$ is a subpath of every path in $\Pi_{X:Y}$. Assume to the contrary that $X \nci W | Z \cup R$ and let $\rho_{X:W}$ be a $(Z \cup R)$-open {\bf route}. Follow $\rho_{X:W}$ until reaching $R$ or $W$, and let $\pi_{X:Y} \in \Pi_{X:Y}$.
\begin{itemize}
\item If $R$ is reached first, then note $R$ must be a collider in $\rho_{X:W}$ for this to be $(Z \cup R)$-open. However, the last occurrence of $R$ in $\rho_{X:W}$ contradicts the assumption that $W$ is not $Z$-connected to $R$ through $Pa(R) \cup Sp(R)$.

\item If $W$ is reached first, then let $\varrho_{R:W}$ denote a $Z$-open path that does not contain any node that is in some path in $\Pi_{X:Y}$ except $R$. Such a path exists by the assumptions in the lemma. Thus, $\rho_{X:W} \cup \varrho_{W:R} \cup \pi_{R:Y}$ is a $Z$-open route, because neither $R$ nor $W$ is a collider in it. The latter follows from the assumption that $Pa(W) \cup Sp(W)=\emptyset$. Moreover, the route does not contain any edge $\oa R$ that is in some path in $\Pi_{X:Y}$ because, otherwise, $\rho_{X:W}$ contains the edge ($\varrho_{W:R}$ and $\pi_{R:Y}$ cannot by definition) and thus it reaches $R$ first. Then, the path corresponding to $\rho_{X:W} \cup \varrho_{W:R} \cup \pi_{R:Y}$ contradicts the assumptions in the lemma.
\end{itemize}
Consequently, $X \ci W | Z \cup R$. Now, assume to the contrary that $Y \nci W | Z \cup R$ and let $\rho_{Y:W}$ be a $(Z \cup R)$-open {\bf route}. Follow $\rho_{Y:W}$ until reaching $R$ or $W$, and let $\pi_{X:Y} \in \Pi_{X:Y}$.
\begin{itemize}
\item If $R$ is reached first, then note $R$ must be a collider in $\rho_{Y:W}$ for this be $(Z \cup R)$-open. However, the last occurrence of $R$ in $\rho_{Y:W}$ contradicts the assumption that $W$ is not $Z$-connected to $R$ through $Pa(R) \cup Sp(R)$.

\item If $W$ is reached first, then let $\varrho_{R:W}$ denote a $Z$-open path that leaves $R$ through $Ch(R)$ and that does not contain any node that is in some path in $\Pi_{X:Y}$ except $R$. Such a path exists by the assumptions in the lemma. Thus, $\pi_{X:R} \cup \varrho_{R:W} \cup \rho_{W:Y}$ is a $Z$-open route, because neither $R$ nor $W$ is a collider in it. The latter follows from the assumption that $Pa(W) \cup Sp(W)=\emptyset$. Moreover, the route does not contain any edge $R \ra$ that is in some path in $\Pi_{X:Y}$ because, otherwise, $\rho_{W:Y}$ contains the edge ($\pi_{X:R}$ and $\varrho_{R:W}$ cannot by definition) and thus $\rho_{Y:W}$ reaches $R$ first. Then, the path corresponding to $\pi_{X:R} \cup \varrho_{R:W} \cup \rho_{W:Y}$ contradicts the assumptions in the lemma.
\end{itemize}
Consequently, $Y \ci W | Z \cup R$. Now, assume to the contrary that $X \nci Y | Z \cup R$ and let $\rho_{X:Y}$ be a $(Z \cup R)$-open {\bf path}. Note that $R$ must be a collider or a descendant of a collider in $\rho_{X:Y}$ because, otherwise, $R$ is not in $\rho_{X:Y}$ and, thus, $\rho_{X:Y}$ is $Z$-open, which contradicts the assumptions in the lemma. Note also that the assumption that $W$ is $Z$-connected to $R$ through $Ch(R)$ implies that some descendant of $R$ is in $Z \cup W$. Actually, some descendant of $R$ must be in $Z$ due to the assumption that $Pa(W) \cup Sp(W)=\emptyset$. Then, $\rho_{X:Y}$ is $Z$-open, which contradicts the assumptions in the lemma. Consequently, $X \ci Y | Z \cup R$. Therefore, the desired result follows from Lemma \ref{lem:aux1}.

Finally, consider the case where $\oa Y = R$ is a subpath of every path in $\Pi_{X:Y}$. Assume to the contrary that $X \nci W | Z \cup Y$ and let $\rho_{X:W}$ be a $(Z \cup Y)$-open {\bf route}. Follow $\rho_{X:W}$ until reaching $Y$ or $W$.
\begin{itemize}
\item If $Y$ is reached first, then note $Y$ must be a collider in $\rho_{X:W}$ for this to be $(Z \cup Y)$-open. However, the last occurrence of $Y$ in $\rho_{X:W}$ contradicts the assumption that $W$ is not $Z$-connected to $R$ through $Pa(Y) \cup Sp(Y)$.

\item If $W$ is reached first, let $\varrho_{Y:W}$ denote a $Z$-open path that does not contain any node that is in some path in $\Pi_{X:Y}$ except $Y$. Such a path exists by the assumptions in the lemma. Then, $\varrho_{X:W} \cup \rho_{W:Y}$ is a $Z$-open route, because $W$ is not a collider in it due to the assumption that $Pa(W) \cup Sp(W)=\emptyset$. Moreover, the route does not contain any edge $\oa Y$ that is in some path in $\Pi_{X:Y}$ because, otherwise, $\rho_{X:W}$ contains the edge ($\varrho_{W:Y}$ cannot by definition) and thus it reaches $Y$ first. Then, the path corresponding to $\varrho_{X:W} \cup \rho_{W:Y}$ contradicts the assumptions in the lemma.
\end{itemize}
Consequently, $X \ci W | Z \cup Y$ and, thus, the desired result follows from Lemma \ref{lem:aux2}.
\end{proof}

\begin{lemma}\label{lem:nonroot4}
Consider a path diagram. Let $\Pi_{X:Y}$ denote all the $Z$-open paths from $X$ to $Y$. Suppose that no path in $\Pi_{X:Y}$ has colliders. Suppose that all the paths in $\Pi_{X:Y}$ have a subpath $\oa R \ra$ or $\oa Y = R$. Let $W$ be a node that is $Z$-connected to $R$ through $Ch(R)$ by a path that does not contain any node that is in some path in $\Pi_{X:Y}$ except $R$. If $Ch(W)= \emptyset$, and $W$ is not $Z$-connected to $R$ by any path that reaches $R$ through $Pa(R) \cup Sp(R)$, and $\Pi_{X:Y}$ are all the $(Z \cup W)$-open paths from $X$ to $Y$, then
\[
\sigma_{XY \cdot ZW} =  \sigma_{XY \cdot Z} \frac{\sigma^2_{R \cdot ZW}}{\sigma^2_{R \cdot Z}}.
\]
\end{lemma}

\begin{proof}
Consider first the case where $\oa R \ra$ is a subpath of every path in $\Pi_{X:Y}$. Assume to the contrary that $X \nci W | Z \cup R$ and let $\rho_{X:W}$ be a $(Z \cup R)$-open {\bf route}. Follow $\rho_{X:W}$ until reaching $R$ or $W$, and let $\pi_{X:Y} \in \Pi_{X:Y}$.
\begin{itemize}
\item If $R$ is reached first, then note that $R$ must be a collider in $\rho_{X:W}$ for this to be $(Z \cup R)$-open. However, the last occurrence of $R$ in $\rho_{X:W}$ contradicts the assumption that $W$ is not $Z$-connected to $R$ through $Pa(R) \cup Sp(R)$.

\item If $W$ is reached first, then let $\varrho_{R:W}$ denote a $Z$-open path that does not contain any node that is in some path in $\Pi_{X:Y}$ except $R$. Such a path exists by the assumptions in the lemma. Thus, $\rho_{X:W} \cup \varrho_{W:R} \cup \pi_{R:Y}$ is a $(Z \cup W)$-open route, because $W$ is a collider in it whereas $R$ is not. The former follows from the assumption that $Ch(W)=\emptyset$. Moreover, the route does not contain any edge $\oa R$ that is in some path in $\Pi_{X:Y}$ because, otherwise, $\rho_{X:W}$ contains the edge ($\varrho_{W:R}$ and $\pi_{R:Y}$ cannot by definition) and thus it reaches $R$ first. Then, the path corresponding to $\rho_{X:W} \cup \varrho_{W:R} \cup \pi_{R:Y}$ contradicts the assumptions in the lemma.
\end{itemize}
Consequently, $X \ci W | Z \cup R$. Now, assume to the contrary that $Y \nci W | Z \cup R$ and let $\rho_{Y:W}$ be a $(Z \cup R)$-open {\bf route}. Follow $\rho_{Y:W}$ until reaching $R$ or $W$, and let $\pi_{X:Y} \in \Pi_{X:Y}$.
\begin{itemize}
\item If $R$ is reached first, then note that $R$ must be a collider in $\rho_{Y:W}$ for this to be $(Z \cup R)$-open. However, the last occurrence of $R$ in $\rho_{Y:W}$ contradicts the assumption that $W$ is not $Z$-connected to $R$ through $Pa(R) \cup Sp(R)$.

\item If $W$ is reached first, then let $\varrho_{R:W}$ denote a $Z$-open path that leaves $R$ through $Ch(R)$ and that does not contain any node that is in some path in $\Pi_{X:Y}$ except $R$. Such a path exists by the assumptions in the lemma. Thus, $\pi_{X:R} \cup \varrho_{R:W} \cup \rho_{W:Y}$ is a $(Z \cup W)$-open route, because $W$ is a collider in it whereas $R$ is not. The former follows from the assumption that $Ch(W)=\emptyset$. Moreover, the route does not contain any edge $R \ra$ that is in some path in $\Pi_{X:Y}$ because, otherwise, $\rho_{W:Y}$ contains the edge ($\pi_{X:R}$ and $\varrho_{R:W}$ cannot by definition) and thus $\rho_{Y:W}$ reaches $R$ first. Then, the path corresponding to $\pi_{X:R} \cup \varrho_{R:W} \cup \rho_{W:Y}$ contradicts the assumptions in the lemma.
\end{itemize}
Consequently, $Y \ci W | Z \cup R$. Now, assume to the contrary that $X \nci Y | Z \cup R$ and let $\rho_{X:Y}$ be a $(Z \cup R)$-open {\bf path}. Note that $R$ must be a collider or a descendant of a collider in $\rho_{X:Y}$ because, otherwise, $R$ is not in $\rho_{X:Y}$ and, thus, $\rho_{X:Y}$ is $Z$-open, which contradicts the assumptions in the lemma. Note also that the assumption that $W$ is $Z$-connected to $R$ through $Ch(R)$ implies that some descendant of $R$ is in $Z \cup W$. Then, $\rho_{X:Y}$ is $(Z \cup W)$-open, which contradicts the assumptions in the lemma. Consequently, $X \ci Y | Z \cup R$. Therefore, the desired result follows from Lemma \ref{lem:aux1}.

Finally, consider the case where $\oa Y = R$ is a subpath of every path in $\Pi_{X:Y}$. Assume to the contrary that $X \nci W | Z \cup Y$ and let $\rho_{X:W}$ be a $(Z \cup Y)$-open {\bf route}. Follow $\rho_{X:W}$ until reaching $Y$ or $W$.
\begin{itemize}
\item If $Y$ is reached first, then note $Y$ must be a collider in $\rho_{X:W}$ for this to be $(Z \cup Y)$-open. However, the last occurrence of $Y$ in $\rho_{X:W}$ contradicts the assumption that $W$ is not $Z$-connected to $R$ through $Pa(Y) \cup Sp(Y)$.

\item If $W$ is reached first, let $\varrho_{Y:W}$ denote a $Z$-open path that does not contain any node that is in some path in $\Pi_{X:Y}$ except $Y$. Such a path exists by the assumptions in the lemma. Then, $\varrho_{X:W} \cup \rho_{W:Y}$ is a $(Z \cup W)$-open route, because $W$ is a collider in it due to the assumption that $Ch(W)=\emptyset$. Moreover, the route does not contain any edge $\oa Y$ that is in some path in $\Pi_{X:Y}$ because, otherwise, $\rho_{X:W}$ contains the edge ($\varrho_{Y:W}$ cannot by definition) and thus it reaches $Y$ first. Then, the path corresponding to $\varrho_{X:W} \cup \rho_{W:Y}$ contradicts the assumptions in the lemma.
\end{itemize}
Consequently, $X \ci W | Z \cup Y$ and, thus, the desired result follows from Lemma \ref{lem:aux2}.
\end{proof}

\begin{proof}[Proof of Theorem \ref{the:condpathroot}]
Consider hereinafter the path diagram conditioned on $S$. We first compute $\sigma_{XY \cdot Z^1}$ from $\sigma_{XY}$ by adding the nodes in $Z^1$ to the conditioning set in the order $W_1, W_2, \ldots$. The assumption that $\Pi_{X:Y}$ are all the $Z$-open paths from $X$ to $Y$ implies that $\Pi_{X:Y}$ are all the $(W_{1:j-1})$-open paths from $X$ to $Y$ because, otherwise, any other path cannot be closed afterwards which contradicts the assumption. To see it, note that a node $W \in Z$ does not close any open path, since there is no subpath $\oa W \ra$ or $\la W \ra$ due to the conditioning operation. Then,
\[
\sigma_{XY \cdot W_1} =  \sigma_{XY} \frac{\sigma^2_{X_1 \cdot W_1}}{\sigma^2_{X_1}}
\]
by Lemma \ref{lem:root1} if $Pa(W_1) \cup Sp(W_1) = \emptyset$, or Lemma \ref{lem:root2} if $Ch(W_1)=\emptyset$. Likewise, 
\[
\sigma_{XY \cdot W_1 W_2} =  \sigma_{XY \cdot W_1} \frac{\sigma^2_{X_1 \cdot W_1 W_2}}{\sigma^2_{X_1 \cdot W_1}}
\]
by Lemma \ref{lem:root1} or \ref{lem:root2}. Combining the last two equations gives
\[
\sigma_{XY \cdot W_1 W_2} =  \sigma_{XY} \frac{\sigma^2_{X_1 \cdot W_1 W_2}}{\sigma^2_{X_1}}.
\]
Continuing with this process for the rest of the nodes in $Z^1$ gives
\begin{equation}\label{eq:Z1a}
\sigma_{XY \cdot Z^1} =  \sigma_{XY} \frac{\sigma^2_{X_1 \cdot Z^1}}{\sigma^2_{X_1}}.
\end{equation}

Now, we compute $\sigma_{XY \cdot Z^1_1}$ from $\sigma_{XY \cdot Z^1}$ by adding the nodes in $Z_1$ to the conditioning set in the order $W_1, W_2, \ldots$. Recall from above that $\Pi_{X:Y}$ are all the $(Z^1 \cup W_{1:j-1})$-open paths from $X$ to $Y$. Then,
\[
\sigma_{XY \cdot Z^1 Z_1} =  \sigma_{XY \cdot Z^1} \frac{\sigma^2_{X_1 \cdot Z^1 Z_1}}{\sigma^2_{X_1 \cdot Z^1}}
\]
by repeating the reasoning that led to Equation \ref{eq:Z1a}. Moreover, combining the last two equations yields
\begin{equation}\label{eq:Z11}
\sigma_{XY \cdot Z^1_1} =  \sigma_{XY} \frac{\sigma^2_{X_1 \cdot Z^1_1}}{\sigma^2_{X_1}}.
\end{equation}

Now, we compute $\sigma_{XY \cdot Z^1_1 Z^2}$ from $\sigma_{XY \cdot Z^1_1}$ by adding the nodes in $Z^2$ to the conditioning set in the order $W_1, W_2, \ldots$. The assumption that there is no $Z$-open route $X_2 \ra A \oo \cdots \oo B \oa X_2$ implies that there is no $(Z^1_1 \cup W_{1:j-1})$-open path from $W_j$ to $X_2$ through $Ch(X_2)$. To see it, assume the opposite. Then, there are $(Z^1_1 \cup W_{1:j-1})$-open paths from $W_j$ to $X_2$ through both $Ch(X_2)$ and $Pa(X_2) \cup Sp(X_2)$. This implies that there is a route $X_2 \ra A \oo \cdots \oo B \oa X_2$ that contains $W_j$, and the route is $(Z^1_1 \cup W_{1:j-1})$-open or $(Z^1_1 \cup W_{1:j})$-open. Then, there is a path $\ra A \oo \cdots \oo B \oa$ that is $(Z^1_1 \cup W_{1:j-1})$-open or $(Z^1_1 \cup W_{1:j})$-open by Lemma \ref{lem:pathroute}. However, this path cannot be closed afterwards which contradicts the assumption. To see it, note that a node $W \in Z$ does not close any open path, since there is no subpath $\oa W \ra$ or $\la W \ra$ due to the conditioning operation. Consequently, there is no $(Z^1_1 \cup W_{1:j-1})$-open path from $W_j$ to $X_2$ through $Ch(X_2)$. Recall also from above that $\Pi_{X:Y}$ are all the $(Z^1_1 \cup W_{1:j-1})$-open paths from $X$ to $Y$. Then,
\[
\sigma_{XY \cdot Z^1_1 W_1} =  \sigma_{XY \cdot Z^1_1}.
\]
by Lemma \ref{lem:nonroot1} if $Pa(W_1) \cup Sp(W_1) = \emptyset$, or Lemma \ref{lem:nonroot2} if $Ch(W_1)=\emptyset$. Likewise,
\[
\sigma_{XY \cdot Z^1_1 W_1 W_2} =  \sigma_{XY \cdot Z^1_1 W_1}.
\]
by Lemma \ref{lem:nonroot1} or \ref{lem:nonroot2}. Combining the last two equations gives
\[
\sigma_{XY \cdot W_1 W_2} =  \sigma_{XY \cdot Z^1_1}.
\]
Continuing with this process for the rest of the nodes in $Z^2$ gives
\begin{equation}\label{eq:Z2a}
\sigma_{XY \cdot Z^1_1 Z^2} =  \sigma_{XY \cdot  Z^1_1}.
\end{equation}

Now, we compute $\sigma_{XY \cdot Z^1_1 Z^2_2}$ from $\sigma_{XY \cdot Z^1_1 Z^2}$ by adding the nodes in $Z_2$ to the conditioning set in the order $W_1, W_2, \ldots$. Recall from above that $\Pi_{X:Y}$ are all the $(Z^{1:2}_1 \cup W_{1:j-1})$-open paths from $X$ to $Y$. Recall also from above that the assumption that there is no $Z$-open route $X_2 \ra A \oo \cdots \oo B \oa X_2$ implies that there is no $(Z^{1:2}_1 \cup W_{1:j-1})$-open path from $W_j$ to $X_2$ through $Pa(X_2) \cup Sp(X_2)$. Then,
\[
\sigma_{XY \cdot Z^1_1 Z^2_2} =  \sigma_{XY \cdot Z^1_1 Z^2} \frac{\sigma^2_{X_2 \cdot Z^1_1 Z^2_2}}{\sigma^2_{X_2 \cdot Z^1_1 Z^2}}
\]
by repeating the reasoning that led to Equation \ref{eq:Z1a} but using Lemmas \ref{lem:nonroot3} and \ref{lem:nonroot4} instead. Moreover, combining the last equation with Equations \ref{eq:Z11} and \ref{eq:Z2a} gives
\[
\sigma_{XY \cdot Z^1_1 Z^2_2} =  \sigma_{XY} \frac{\sigma^2_{X_1 \cdot Z^1_1}}{\sigma^2_{X_1}} \frac{\sigma^2_{X_2 \cdot Z^1_1 Z^2_2}}{\sigma^2_{X_2 \cdot Z^1_1 Z^2}}.
\]

Finally, continuing with the process above for $X_3, \ldots, X_{m+n}$ yields
\[
\sigma_{X Y \cdot Z^{1:m+n}_{1:m+n}} = \sigma_{X Y} \frac{\sigma^2_{X_1 \cdot Z_1^1}}{\sigma^2_{X_1}} \prod_{i=2}^{m+n} \frac{\sigma^2_{X_i \cdot Z_{1:i}^{1:i}}}{\sigma^2_{X_i \cdot Z_{1:i-1}^{1:i}}}
\]
which implies the desired result by repeated application of Lemma \ref{lem:aux3}.
\end{proof}

\begin{proof}[Proof of Theorem \ref{the:condpathnonroot}]
The proof is analogous to that of Theorem \ref{the:condpathroot}.
\end{proof}